\documentclass[12pt,a4paper]{amsart}

\usepackage{amsmath,amsfonts,amssymb,color}
\usepackage[hmargin=2cm,vmargin=2cm]{geometry}
\usepackage{color,xcolor}
\usepackage[hypertexnames=false,hyperfootnotes=false,colorlinks=true,linkcolor=blue,%
citecolor=purple,filecolor=magenta,urlcolor=cyan,unicode,linktocpage=true,pagebackref=false]{hyperref}
\usepackage{nameref,zref-xr}                    
\usepackage{amsthm}
\usepackage{cancel}
\usepackage{enumerate}

\newcommand{\p}{{\partial}}
\newcommand{\mbC}{\mathbb C}
\newcommand{\oM}{\overline{\mathcal M}}
\def\mbQ{{\mathbb Q}}
\def\d{{\partial}}

\newcommand{\eps}{\varepsilon}
\newcommand{\Coef}{\mathrm{Coef}}

\newcommand{\mcF}{\mathcal{F}}
\newcommand{\mcL}{\mathcal{L}}

\newcommand{\tD}{\widetilde{D}}

\newcommand{\ext}{\mathrm{ext}}

\newcommand{\ov}{{\overline{v}}}

\newcommand{\End}{\mathrm{End}}
\newcommand{\Id}{\mathrm{Id}}
\newcommand{\mfg}{\mathfrak{g}}
\newcommand{\oF}{\overline{F}}

\newcommand{\ad}{\mathrm{ad}}
\newcommand{\omcF}{\overline{\mathcal{F}}}
\newcommand{\tr}{\widetilde{r}}
\newcommand{\ts}{\widetilde{s}}

\newcommand{\GL}{{\mathrm{GL}}}

\newcommand{\ob}{\overline{b}}

\newcommand{\ot}{\overline{t}}
\newcommand{\orig}{\mathrm{orig}}
\newcommand{\mcR}{\mathcal{R}}
\newcommand{\mbZ}{\mathbb{Z}}
\newcommand{\tG}{\widetilde{G}}
\newcommand{\tmfg}{\widetilde{\mathfrak{g}}}
\newcommand{\tPsi}{\widetilde{\Psi}}
\newcommand{\tGamma}{\widetilde{\Gamma}}

\newcommand{\diag}{\mathrm{diag}}
\newcommand{\cF}{\mathcal{F}}
\newcommand{\tM}{\widetilde{M}}

\newcommand{\oA}{\overline{A}}
\newcommand{\KW}{\mathrm{KW}}
\newcommand{\PST}{\mathrm{PST}}
\newcommand{\tmcF}{\widetilde{\mathcal{F}}}
\newcommand{\tGL}{\widetilde{\mathrm{GL}}}
\newcommand{\cL}{\mathcal{L}}
\newcommand{\otheta}{\overline{\theta}}
\newcommand{\Desc}{\mathrm{Desc}}
\newcommand{\Anc}{\mathrm{Anc}}

\newcommand{\tc}{\widetilde{c}}
\newcommand{\lb}{\left(}
\newcommand{\rb}{\right)}
\newcommand{\cub}{\mathrm{cub}}
\newcommand{\quadr}{\mathrm{quad}}
\newcommand{\mcP}{\mathcal{P}}
\newcommand{\anc}{\mathrm{anc}}
\newcommand{\desc}{\mathrm{desc}}
\newcommand{\tE}{\widetilde{E}}
\newcommand{\un}{{1\!\! 1}}

\newcommand{\tv}{\widetilde{v}}
\newcommand{\tphi}{\widetilde{\phi}}

\newtheorem{theorem}{Theorem}
\newtheorem{proposition}{Proposition}[section]
\newtheorem{lemma}[proposition]{Lemma}

\newtheorem{example}[proposition]{Example}
\newtheorem{remark}[proposition]{Remark}

\numberwithin{equation}{section}

\begin{document}

\title{A construction of open descendant potentials in all genera}

\author{Alexander Alexandrov}
\address[A. Alexandrov]{Center for Geometry and Physics, Institute for Basic Science (IBS), Pohang 37673, Korea
}
\email{alexandrovsash@gmail.com}

\author{Alexey Basalaev}
\address[A. Basalaev]{Faculty of Mathematics, National Research University Higher School of Economics, 6 Usacheva str., Moscow, 119048, Russian Federation;\smallskip\newline 
Center for Advanced Studies, Skolkovo Institute of Science and Technology, 1 Nobel str., Moscow, 143026, Russian Federation}
\email{abasalaev@hse.ru}

\author{Alexandr Buryak}
\address[A. Buryak]{Faculty of Mathematics, National Research University Higher School of Economics, 6 Usacheva str., Moscow, 119048, Russian Federation;\smallskip\newline 
Center for Advanced Studies, Skolkovo Institute of Science and Technology, 1 Nobel str., Moscow, 143026, Russian Federation;\smallskip\newline
Faculty of Mechanics and Mathematics, Lomonosov Moscow State University, \newline 
GSP-1, 119991 Moscow, Russian Federation}
\email{aburyak@hse.ru}

\begin{abstract}
We present a construction of an open analogue of total descendant and total ancestor potentials via an ``open version'' of Givental's action. Our construction gives a genus expansion for an arbitrary solution to the open WDVV equations satisfying a semisimplicity condition and admitting a unit. We show that the open total descendant potentials we define satisfy the open topological recursion relations in genus $0$ and $1$, the open string and open dilaton equations. We finish the paper with a computation of the simplest nontrivial open correlator in genus $1$ using our construction.
\end{abstract}

\date{\today}

\maketitle

\setcounter{tocdepth}{1}
\tableofcontents

\section{Introduction}

\subsection{The WDVV equations} 

The {\it WDVV equations}, also called the {\it associativity equations}, is a system of nonlinear PDEs for one function depending on a finite number of variables. Let $N\ge 1$ and $\eta=(\eta_{\alpha\beta})$ be an~$N\times N$ symmetric nondegenerate matrix with complex coefficients. The WDVV equations is the following system of PDEs for an analytic function $F(t^1,\ldots,t^N)$ defined on some open subset $M\subset\mbC^N$:
\begin{gather}\label{eq:WDVV equations}
\frac{\d^3 F}{\d t^\alpha\d t^\beta\d t^\mu}\eta^{\mu\nu}\frac{\d^3 F}{\d t^\nu\d t^\gamma\d t^\delta}=\frac{\d^3 F}{\d t^\delta\d t^\beta\d t^\mu}\eta^{\mu\nu}\frac{\d^3 F}{\d t^\nu\d t^\gamma\d t^\alpha},\quad 1\le\alpha,\beta,\gamma,\delta\le N,
\end{gather}
where $(\eta^{\alpha\beta}):=\eta^{-1}$ and we use the convention of sum over repeated Greek indices. Equations~\eqref{eq:WDVV equations} are equivalent to the fact that the tensor $c^\alpha_{\beta\gamma}:=\eta^{\alpha\mu}\frac{\d^3 F}{\d t^\mu\d t^\beta\d t^\gamma}$ defines the structure of an associative algebra in the tangent bundle $TM$. We will consider the case when this algebra structure has a unit given by a vector field~$e=A^\alpha\frac{\d}{\d t^\alpha}$, $A^1,\ldots,A^N\in\mbC$. This condition can be equivalently written as 
\begin{gather}\label{eq:unit for WDVV}
A^\mu\frac{\d^3 F}{\d t^\mu\d t^\alpha\d t^\beta}=\eta_{\alpha\beta}.
\end{gather}
In this case we will say that the solution $F$ to the WDVV equations \emph{admits the unit $A^\alpha\frac{\d}{\d t^\alpha}$}. Such a function~$F$ defines the structure of a {\it Dubrovin--Frobenius manifold} on~$M$ and is also called the \textit{Dubrovin--Frobenius manifold potential}. This structure appears in different areas of mathematics, including curve counting theories in algebraic geometry (Gromov--Witten theory, Fan--Jarvis--Ruan--Witten theory) and singularity theory. A systematic study of Dubrovin--Frobenius manifolds was first done by Dubrovin~\cite{Dub96,Dub99}. The matrix $\eta$ is often called the \emph{metric}.

\medskip

We will say that a solution $F$ to the WDVV equations is \emph{semisimple} at a point $p\in M$, if the algebra structure on $T_p M$ does not have nilpotents. Such solutions to the WDVV equations play a special role in Givental's theory. Note that the function $F$ comes in equations~\eqref{eq:WDVV equations} and~\eqref{eq:unit for WDVV} together with the third derivatives. So we will consider solutions to the WDVV equations up to adding a quadratic polynomial in $t^1,\ldots,t^N$.

\medskip

We will often call the WDVV equations the \emph{closed} WDVV equations in order to distinguish them from the open WDVV equations that will appear below. Correspondingly we will add the superscript ``c'' to the notation for solutions to~\eqref{eq:WDVV equations}.

\medskip

\subsection{The Givental theory} 

In our paper, we will mostly consider the case when $M$ is a formal neighbourhood of some point $\ot_{\orig}=(t^1_\orig,\ldots,t^N_\orig)\in\mbC^N$ meaning that we will consider solutions to the closed WDVV equations of the form $F^c\in\mbC[[t^1-t^1_\orig,\ldots,t^N-t^N_\orig]]$. Let us introduce the notation
$$
\mcR^{\ot_{\orig}}_N:=\mbC[[t^1-t^1_\orig,\ldots,t^N-t^N_\orig]].
$$

\medskip

In \cite{Giv04,Giv01b} (see also~\cite{L05}) Givental introduced a group $G^c_{N,+}$ acting on the space of all solutions $F^c\in\mbC[[t^1,\ldots,t^N]]$ to the closed WDVV equations admitting a unit. The action does not change the algebra structure on $T_0\mbC^N$, and the action is transitive on the subspace of solutions defining a fixed semisimple algebra structure on $T_0\mbC^N$. Therefore, any solution to the closed WDVV equations, admitting a unit and that is semisimple at $0$, can be obtained from the solution $F^c=\sum_{i=1}^N a_i\frac{(t^i)^3}{6}$, with $\eta_{ij}=\delta_{ij}$ and $e=\sum_{i=1}^N a_i^{-1}\frac{\d}{\d t^i}$, $a_1,\ldots,a_N\in\mbC^*$, making a linear change of variables $t^\alpha\mapsto M^\alpha_\beta t^\beta$, $M=(M^\alpha_\beta)\in\GL(\mbC^N)$, and acting by an element of the group~$G^c_{N,+}$.

\medskip 

Let $t_0,t_1,\ldots$ and $\eps$ be formal variables and consider the Kontsevich--Witten potential~\cite{K92,Wit91} 
$$
\mcF^\KW(t_0,t_1,\ldots,\eps)=\sum_{g\ge 0}\eps^{2g-2}\mcF^\KW_g(t_0,t_1,\ldots),
$$
where
$$
\mcF^\KW_g(t_0,t_1,\ldots):=\sum_{n\ge 1,\,2g-2+n>0}\left(\int_{\oM_{g,n}}\prod_{i=1}^n\psi_i^{d_i}\right)\frac{\prod_{i=1}^n t_{d_i}}{n!}\in\mbC[[t_0,t_1,\ldots]].
$$
Here $\oM_{g,n}$ is the moduli space of stable algebraic curves of genus $g$ with $n$ marked points, and $\psi_i\in H^2(\oM_{g,n},\mbQ)$, $1\le i\le n$, is the first Chern class of the $i$-th tautological line bundle over~$\oM_{g,n}$ formed by the cotangent lines at the $i$-th marked point on stable curves from $\oM_{g,n}$. Note that $\left.\mcF^\KW_0\right|_{t_{\ge 1}=0}=\frac{t_0^3}{6}$.

\medskip

Givental postulated the function $\sum_{i=1}^N\mcF^\KW(a_i t^i_*,a_i\eps)$ to be a \emph{closed total ancestor potential} associated to the solution $F^c=\sum_{i=1}^N a_i\frac{(t^i)^3}{6}$ to the closed WDVV equations. He constructed a $G^c_{N,+}$-action on a certain space of formal series of the form
\begin{gather}\label{eq:form of closed potential}
\mcF^c(t^*_*,\eps)=\sum_{g\ge 0}\eps^{2g-2}\mcF^c_g(t^*_*),\quad \mcF^c_g\in\mbC[[t^*_*]],
\end{gather}
where $t^\alpha_0=t^\alpha, t^\alpha_1, t^\alpha_2,\ldots$, $1\le\alpha\le N$, and $\eps$ are formal variables. Acting by a linear change of variables and then by an element of $G^c_{N,+}$ on the function $\sum_{i=1}^N\mcF^\KW(a_i t^i_*,a_i\eps)$, Givental~\cite{Giv01b} constructed a closed total ancestor potential $\mcF^{c,\anc}(t^*_*,\eps)=\sum_{g\ge 0}\eps^{2g-2}\mcF^{c,\anc}_g(t^*_*)$ associated to an arbitrary solution $F^c\in\mbC[[t^1,\ldots,t^N]]$ to the closed WDVV equations, admitting a unit and that is semisimple at $0\in\mbC^N$. The closed total ancestor potential $\mcF^{c,\anc}$ satisfies $\left.\mcF^{c,\anc}_0\right|_{t^*_{\ge 1}=0}=F^c$. Note that there is actually an infinite dimensional space of closed total ancestor potentials associated to a given solution to the closed WDVV equations.

\medskip

Givental also introduced another group $G^c_{N,-}$ acting on solutions $F^c$ to the closed WDVV equations (admitting a unit) of the form $F^c\in\mcR^{\ot_\orig}_N$, $\ot_\orig\in\mbC^N$, endowed with a certain additional structure called a \emph{calibration}. The $G^c_{N,-}$-action essentially acts only on the calibration. Using this action, starting from the closed total ancestor potentials Givental introduced the space of \emph{closed total descendant potentials} $\mcF^{c,\desc}$ of the form
$$
\mcF^{c,\desc}=\sum_{g\ge 0}\eps^{2g-2}\mcF^{c,\desc}_g,\quad\mcF^{c,\desc}_g\in\mcR^{\ot_{\orig}}_N[[t^*_{\ge 1}]],\quad \ot_\orig\in\mbC^N.
$$
One can see this construction as an axiomatization of the reconstruction formula for the total descendant potential starting from the total ancestor potential in the Gromov--Witten theory of some target variety~$X$.

\medskip

Remarkably, by a result of Teleman~\cite{Tel12}, if a semisimple solution $F^c$ to the closed WDVV equations is equal to the generating series of primary genus~$0$ Gromov--Witten invariants of some target variety $X$, one of Givental's closed total descendant potentials associated to $F^c$ coincides with the generating series of all (descendant and all genera) Gromov--Witten invariants of~$X$. There is a way to fix the closed total descendant potential uniquely using a certain homogeneity condition, but we will not discuss that in the paper. Actually, the Teleman result is not limited to Gromov--Witten theory and works in a much more general framework, where the Gromov--Witten potential is replaced by the total ancestor potential associated to a semisimple cohomological field theory~\cite{KM94}.

\medskip

\subsection{The open WDVV equations}\label{subsection:open WDVV}

More recently, a remarkable system of PDEs, extending the closed WDVV equations~\eqref{eq:WDVV equations}, appeared in the literature. For a fixed solution $F^c$ to the closed WDVV equations \eqref{eq:WDVV equations}, the {\it open WDVV equations} are the following PDEs for an analytic function $F^o(t^1,\ldots,t^{N+1})$ defined on an open subset $\tM\subset\mbC^{N+1}=\mbC^N\times\mbC$ of the form $\tM=M\times U$:
\begin{gather}\label{eq:open WDVV}
\frac{\d^3F^c}{\d t^\alpha\d t^\beta\d t^\mu}\eta^{\mu\nu}\frac{\d^2F^o}{\d t^\nu\d t^\gamma}+\frac{\d^2F^o}{\d t^\alpha\d t^\beta}\frac{\d^2F^o}{\d t^{N+1}\d t^\gamma}=\frac{\d^3F^c}{\d t^\gamma\d t^\beta\d t^\mu}\eta^{\mu\nu}\frac{\d^2F^o}{\d t^\nu\d t^\alpha}+\frac{\d^2F^o}{\d t^\gamma\d t^\beta}\frac{\d^2F^o}{\d t^{N+1}\d t^\alpha}, 
\end{gather}
where $1\le\alpha,\beta,\gamma\le N+1$. These equations appear in open Gromov--Witten theory~\cite{HS12,ST19,CZ21}, open $r$-spin theory~\cite{BCT18} (see~\cite[Section~4]{Bur20} explaining that the open topological recursion relations in genus $0$ appearing in~\cite{BCT18} imply the open WDVV equations), closed extended $r$-spin theory~\cite{BCT19}, and also in~\cite{BB21a} in the context of open Saito theory. The system of equations~\eqref{eq:WDVV equations} and~\eqref{eq:open WDVV} is equivalent to the fact that the tensor 
\[
\tc^\alpha_{\beta\gamma}:=\begin{cases}\eta^{\alpha\mu}\frac{\d^3 F^c}{\d t^\mu\d t^\beta\d t^\gamma},&\text{if $1\le\alpha\le N$},\\\frac{\d^2 F^o}{\d t^\beta\d t^\gamma},&\text{if $\alpha=N+1$},\end{cases}
\]
defines the structure of an associative algebra in the tangent bundle $T\tM$. 

\medskip

We will say that a pair $(F^c,F^o)$ of solutions to equations~\eqref{eq:WDVV equations} and~\eqref{eq:open WDVV} \emph{admits a unit} if the algebra structure defined by $\tc^\alpha_{\beta\gamma}$ has a unit~$e$ of the form $e=A^\alpha \frac{\d}{\d t^\alpha}$, $A^1,\ldots,A^{N+1}\in\mbC$. Then clearly the vector field $\sum_{\alpha=1}^N A^\alpha\frac{\d}{\d t^\alpha}$ gives a unit for the algebra structure in $TM$ defined by $F^c$. The semisimplicity condition at some point~$p\in\tM$ means that the algebra structure on~$T_p\tM$ does not have nilpotents. We will consider solutions to the open WDVV equations up to adding a linear polynomial in $t^1,\ldots,t^{N+1}$.

\medskip

\subsection{The main results}\label{subsection:main results}

The system of open WDVV equations extends the system of closed WDVV equations, so it is natural to ask whether there exists a corresponding extension of the Givental theory. The construction of this extension is the main result of the paper.

\medskip

We will introduce a group $G^o_{N+1,+}$, with a natural projection $G^o_{N+1,+}\to G^c_{N,+}$, acting on the space of pairs of solutions $(F^c,F^o)$, $F^c,F^o\in\mbC[[t^1,\ldots,t^{N+1}]]$, to equations~\eqref{eq:WDVV equations} and~\eqref{eq:open WDVV} admitting a unit. This action is an extension of the Givental $G^c_{N,+}$-action. The action does not change the algebra structure on $T_0\mbC^{N+1}$, and we will prove that the action is transitive on the subspace of solutions defining a fixed semisimple algebra structure on $T_0\mbC^{N+1}$. 

\medskip

Additionally to $t_0,t_1,\ldots$ and $\eps$, consider formal variables $s_0=s,s_1,\ldots$ and consider the Pandharipande--Solomon--Tessler \cite{PST14,Tes15} generating series of intersection numbers on the moduli spaces of Riemann surfaces with boundary $\oM_{g,k,l}$:
$$
\mcF^\PST(t_*,s_*,\eps) = \sum_{g \ge 0} \eps^{g-1} \mcF_g^\PST(t_*,s_*),\quad \mcF_g^\PST\in\mbC[[t_*,s_*]].
$$

\begin{remark}\label{remark:descendant of s}
To be precise, Pandharipande, Solomon, and Tessler introduced the function $\left.\mcF^\PST\right|_{s_{\ge 1}=0}$, from which the full function $\mcF^\PST$, introduced in~\cite[Section~1.5]{Bur15}, can be reconstructed using the relation 
\begin{gather}\label{eq:descendants of s in PST}
\frac{\d\exp\left(\mcF^\PST\right)}{\d s_n}=\frac{\eps^n}{(n+1)!}\frac{\d^{n+1}\exp\left(\mcF^\PST\right)}{\d s^{n+1}}.
\end{gather}
Introducing the dependance on the variables~$s_{\ge 1}$ is very natural from the point of view of integrable systems, matrix models~\cite{Ale15}, and Virasoro constraints~\cite[Section 1.4]{Bur16}.
\end{remark}

\noindent Note that the functions 
\[
F^c(t^1)=\left.\mcF^\KW_0(t^1_*)\right|_{t^1_{\ge 1}=0}=\frac{(t^1)^3}{6}, \qquad F^o(t^1,t^2)=\left.\mcF^\PST_0(t^1_*,t^2_*)\right|_{t^*_{\ge 1}=0} = t^1t^2 + \frac{(t^2)^3}{6}, 
\]
satisfy the system of closed and open WDVV equations with $\eta=1$ and admit the unit $e=\frac{\d}{\d t^1}$. More generally, the functions
$$
F^c(t^1,\ldots,t^N)=\sum_{i=1}^N a_i\frac{(t^i)^3}{6},\qquad F^o(t^1,\ldots,t^{N+1})=a_1t^1t^{N+1}+ \frac{(t^{N+1})^3}{6},
$$
satisfy the system of closed and open WDVV equations with $\eta_{ij}=\delta_{ij}$ and admit the unit $\sum_{i=1}^N a_i^{-1}\frac{\d}{\d t^i}$, $a_i\in\mbC^*$. 

\medskip

We will introduce a $G^o_{N+1,+}$-action on a certain space of pairs $(\mcF^c,\mcF^o)$, where $\mcF^c$ has the form~\eqref{eq:form of closed potential}, and $\mcF^o=\mcF^o(t^1_*,\ldots,t^{N+1}_*,\eps)$ has the form
$$
\mcF^o(t^1_*,\ldots,t^{N+1}_*,\eps)=\sum_{g\ge 0}\eps^{g-1}\mcF^o_g(t^1_*,\ldots,t^{N+1}_*),\quad \mcF^o_g(t^1_*,\ldots,t^{N+1}_*)\in\mbC[[t^1_*,\ldots,t^{N+1}_*]].
$$ 
Applying to the pair $\left(\sum_{i=1}^N \mcF^{KW}(a_i t^i_*,a_i\eps),\mcF^\PST(a_1 t^1_*,t^{N+1}_*,\eps)\right)$ a certain shift, a linear change of variables, and the $G^o_{N+1,+}$-action, we obtain a space of pairs $(\mcF^{c,\anc},\mcF^{o,\anc})$, where $\mcF^{c,\anc}$ is a closed total ancestor potential, and $\mcF^{o,\anc}$ is a new function, which we call an \emph{open total ancestor potential}.

\medskip

Similarly to the Givental construction, we will introduce a group $G^o_{N+1,-}$, and acting by it on pairs $(\mcF^{c,\anc},\mcF^{o,\anc})$ we obtain a space of pairs $(\mcF^{c,\desc},\mcF^{o,\desc})$, where $\mcF^{c,\desc}$ is a closed total descendant potential and $\mcF^{o,\desc}$ is a new function of the form
\begin{gather*}
\mcF^{o,\desc}=\sum_{g\ge 0}\eps^{g-1}\mcF^{o,\desc}_g,\quad \mcF^{o,\desc}_g\in\mcR^{\ot_\orig}_{N+1}[[t^*_{\ge 1}]],\quad \ot_\orig\in\mbC^{N+1},
\end{gather*}
which we call an \emph{open total descendant potential}. Note that then $\mcF^{c,\desc}_g\in\mcR^{\pi(\ot_\orig)}_N[[t^*_{\ge 1}]]$, where $\pi\colon\mbC^{N+1}\to\mbC^N$ is the projection to the first $N$ coordinates. Here is our first result.

\medskip

\begin{theorem}\label{theorem:main}
\begin{enumerate}[ 1.]

\item For any pair $(\mcF^{c,\desc},\mcF^{o,\desc})$ of total descendant potentials, the pair $(F^c,F^o)$, where $F^c:=\left.\mcF^{c,\desc}_0\right|_{t^*_{\ge 1}=0}$ and $F^o:=\left.\mcF^{o,\desc}_0\right|_{t^*_{\ge 1}=0}$, satisfies the system of closed and open WDVV equations, admits a unit, and is semisimple at $\ot_\orig\in\mbC^{N+1}$. 

\smallskip

\item For any pair $(F^c,F^o)$ of solutions to the closed and open WDVV equations having the form $F^c\in\mcR^{\pi(\ot_\orig)}_N$, $F^o\in\mcR^{\ot_\orig}_{N+1}$, $\ot_\orig\in\mbC^{N+1}$, admitting a unit, and that is semisimple at $\ot_\orig\in\mbC^{N+1}$, there exists a pair $(\mcF^{c,\desc},\mcF^{o,\desc})$ of total descendant potentials such that $\left.F^c=\mcF^{c,\desc}_0\right|_{t^*_{\ge 1}=0}$ and $F^o=\left.\mcF^{o,\desc}_0\right|_{t^*_{\ge 1}=0}$.
\end{enumerate}
\end{theorem}

\medskip

We believe that the space of our open total descendant potentials contains the (appropriately defined) open Gromov--Witten potentials of smooth projective varieties and open FJRW potentials, satisfying the semisimplicity condition, which are rigorously constructed in genera higher than zero only in a limited class of cases. Nevertheless, for such geometrically defined potentials $\mcF^o = \sum_{g \ge 0} \eps^{g-1} \mcF^o_g$, which should come with a closed total descendant potential $\mcF^c=\sum_{g\ge 0}\eps^{2g-2}\mcF^c_g$, it is believed that $F^o:=\left.\mcF^o_0\right|_{t^*_{\ge 1}=0}$ should satisfy the open WDVV equations and admit a unit $e=A^\alpha\frac{\d}{\d t^\alpha}$, and that $\mcF^o$ should satisfy the following PDEs:
\begin{itemize}
\item the \emph{open topological recursion relations in genus $0$} (open TRR-$0$ relations)
\begin{equation}\label{eq:open TRR0}
\frac{\d^2\mcF_0^o}{\d t^\alpha_{a+1}\d t^\beta_b}=\frac{\d^2\mcF_0^c}{\d t^\alpha_a\d t^\mu_0}\eta^{\mu \nu}\frac{\d^2\mcF_0^o}{\d t^\nu_0\d t^\beta_b}+\frac{\d\mcF^o_0}{\d t^\alpha_a}\frac{\d^2\mcF^o_0}{\d t_0^{N+1}\d t^\beta_b}, \quad 1\le\alpha,\beta\le N+1, \quad a,b\ge 0;
\end{equation}

\smallskip

\item the \emph{open topological recursion relations in genus $1$} (open TRR-$1$ relations)
\begin{gather}\label{eq:open TRR1}
\frac{\d\mcF_1^o}{\d t^\alpha_{a+1}}=\frac{\d^2\mcF^c_0}{\d t^\alpha_a\d t^\mu_0}\eta^{\mu\nu}\frac{\d\mcF^o_1}{\d t^\nu_0}+\frac{\d\mcF^o_0}{\d t^\alpha_a}\frac{\d\mcF^o_1}{\d t_0^{N+1}}+\frac{1}{2}\frac{\d^2\mcF_0^o}{\d t^\alpha_a\d t_0^{N+1}},\quad 1\le\alpha\le N+1,\quad a\ge 0;
\end{gather}

\smallskip

\item the \emph{open string equation}
\begin{equation}\label{eq:open string equation}
A^\alpha\frac{\d\mcF^o}{\d t^\alpha_0}=\sum_{d\ge 0}t^\alpha_{d+1}\frac{\d\mcF^o}{\d t^\alpha_d}+\eps^{-1}t^{N+1}_0;
\end{equation}

\smallskip

\item the \emph{open dilaton equation}
\begin{equation}\label{eq:open dilaton equation}
A^\alpha\frac{\d\mcF^o}{\d t^\alpha_1}=\sum_{d\ge 0}t^\alpha_{d}\frac{\d\mcF^o}{\d t^\alpha_d}+\eps\frac{\d\mcF^o}{\d\eps}+\frac{1}{2}.
\end{equation}
\end{itemize}

\medskip

\begin{theorem}\label{theorem:main2}
For any pair of total descendant potentials $(\mcF^{c,\desc},\mcF^{o,\desc})$ equations~\eqref{eq:open TRR0}, \eqref{eq:open TRR1}, \eqref{eq:open string equation}, and~\eqref{eq:open dilaton equation} are satisfied.
\end{theorem}

\medskip

Finally, in Section~\ref{section:more properties}, we discuss in more details the action of the group $G^o_{N+1,+}$ on the space of pairs of closed and open total ancestor potentials and compute explicitly the simplest nontrivial correlator in an arbitrary open ancestor potential in genus~$1$.

\medskip

\subsection{Notation and conventions}
 
\begin{itemize}

\item As we already mentioned before, throughout the text we use the Einstein summation convention for repeated upper and lower Greek indices. Also, tensors of different ranks will often simultaneously appear in some formulas. In this case, we assume that a summation is performed in the range of indices where all the terms in the sum are well defined.

\smallskip

\item When it does not lead to a confusion, we use the symbol $*$ to indicate any value, in the appropriate range, of a sub- or superscript.

\smallskip

\item For a Lie group $G$ acting on a smooth manifold $M$, we denote by $g.x\in M$ the result of the action of $g\in G$ on a point $x\in M$. If $\mathfrak{g}$ is the Lie algebra corresponding to $G$ and $h\in \mathfrak{g}$, then we denote by $h.x\in T_x M$ the result of the infinitesimal action of $h$ on $x$.

\end{itemize}

\medskip

\subsection{Acknowledgements}
The work of Alexander Alexandrov was supported by the Institute for Basic Science (IBS-R003-D1). The work of Alexey Basalaev (Sections~\ref{subsection:main results} and~\ref{section:proof of theorem-2}) was supported by International Laboratory of Cluster Geometry NRU HSE, RF Government grant, ag. no. 075-15-2021-608 dated 08.06.2021. The work of Alexandr Buryak (Sections~\ref{subsection:open WDVV},~\ref{section:open Givental}, and~\ref{section:more properties}) was supported by the grant no.~20-11-20214 of the Russian Science Foundation.

\medskip

We thank the anonymous referee of the paper for the remarks that led to a substantial improvement of the paper.

\medskip


\section{A Givental-type theory for solutions to the open WDVV equations}\label{section:open Givental}

The goal of this section is to introduce the space of open descendant potentials $\mcF^o_0$ that will be the genus $0$ part of the open total descendant potentials that we will define in the next section. We will also introduce groups $G^o_{N+1,\pm}$ together with their actions on appropriate subspaces of the space of open descendant potentials. This will be done as a reduction of the Givental-type theory for flat F-manifolds developed in~\cite{ABLR20}.

\medskip

\subsection{Flat F-manifolds}

Let $L\ge 1$. The \emph{oriented WDVV equations} are the following PDEs for $L$ analytic functions $F^1,\ldots,F^L$ on an open subset $M\subset\mbC^L$:
\begin{gather}\label{eq:oriented WDVV}
\frac{\d^2 F^\alpha}{\d t^\beta \d t^\mu} \frac{\d^2 F^\mu}{\d t^\gamma \d t^\delta} = \frac{\d^2 F^\alpha}{\d t^\gamma \d t^\mu} \frac{\d^2 F^\mu}{\d t^\beta \d t^\delta}, \quad 1\leq \alpha,\beta,\gamma,\delta\leq L.
\end{gather}
The functions $F^\alpha$ will be considered up to adding a linear polynomial in $t^1,\ldots,t^L$. Equations~\eqref{eq:oriented WDVV} are equivalent to the fact that the tensor $c^\alpha_{\beta\gamma}:=\frac{\d^2 F^\alpha}{\d t^\beta\d t^\gamma}$ defines the structure of an associative algebra in the tangent bundle $TM$. Suppose that this algebra structure has a unit~$e$ of the form $e=A^\alpha\frac{\d}{\d t^\alpha}$, $A^\alpha\in\mbC$. Then the $L$-tuple of functions $\oF=(F^1,\ldots,F^L)$ defines the structure of a \emph{flat F-manifold} on $M$ and is called the {\it vector potential} of this flat F-manifold. 

\medskip

Note that if $L=N$ and $F^\alpha=\eta^{\alpha\mu}\frac{\d F^c}{\d t^\mu}$, where $\eta=(\eta_{\alpha\beta})$ is an $N\times N$ symmetric nondegenerate matrix and $F^c=F^c(t^1,\ldots,t^N)$ is an analytic function, then equations~\eqref{eq:oriented WDVV} are equivalent to the closed WDVV equations for the function $F^c$.

\medskip

By an observation of Paolo Rossi, if $L=N+1$ and $F^c=F^c(t^1,\ldots,t^N)$, $F^o=F^o(t^1,\ldots,t^{N+1})$ are two analytic functions, then equations \eqref{eq:oriented WDVV} for $F^\alpha:=\eta^{\alpha\mu}\frac{\d F^c}{\d t^\mu}$, $1\le\alpha\le N$, $F^{N+1}:=F^o$ are equivalent to the system of closed and open WDVV equations for the pair $(F^c,F^o)$.

\medskip

\subsection{Descendant vector potentials of flat F-manifolds}

Let us fix $L\ge 1$, a nonzero vector $\overline{A}=(A^1,\ldots,A^L)\in\mbC^L$, and a point $\ot_\orig=(t^1_\orig,\ldots,t^L_\orig)\in\mbC^L$.

\medskip

In \cite[Section~2.1]{ABLR20} the authors introduced the notion of a sequence of descendant vector potentials of a flat F-manifold. An equivalent description is given by~\cite[Proposition~2.2]{ABLR20}: a sequence of $L$-tuples of functions~$\omcF^a_0=(\mcF^{1,a}_0,\ldots,\mcF^{L,a}_0)$, $\mcF^{\alpha,a}_0\in\mcR_L^{\ot_{\orig}}[[t^*_{\ge 1}]]$, $a\ge 0$, is a sequence of {\it descendant vector potentials} of a flat F-manifold if and only if the following equations are satisfied:
\begin{align}
\sum_{b\ge 0}q^\beta_{b+1}\frac{\d\mcF^{\alpha,a}_0}{\d q^\beta_b}=&-\mcF^{\alpha,a-1}_0,&& 1\le\alpha\le L,\quad a\in\mbZ,\label{eq:descendant vector potentials-1}\\
\sum_{b\ge 0}q^\beta_b\frac{\d\mcF_0^{\alpha,a}}{\d q^\beta_b}=&\mcF_0^{\alpha,a},&& 1\le\alpha\le L,\quad a\in\mbZ,\label{eq:descendant vector potentials-2}\\
\frac{\d^2\mcF_0^{\alpha,0}}{\d q^\beta_{b+1}\d q^\gamma_c}=&\frac{\d\mcF_0^{\mu,0}}{\d q^\beta_b}\frac{\d^2\mcF_0^{\alpha,0}}{\d q^\mu_0\d q^\gamma_c},&& 1\le\alpha,\beta,\gamma\le L,\quad b,c\ge 0,\label{eq:descendant vector potentials-3}\\
\frac{\d\mcF_0^{\alpha,a+1}}{\d q^\beta_b}+\frac{\d\mcF_0^{\alpha,a}}{\d q^\beta_{b+1}}=&\frac{\d\mcF_0^{\alpha,a}}{\d q^\mu_0}\frac{\d\mcF_0^{\mu,0}}{\d q^\beta_b},&& 1\le\alpha,\beta\le L,\quad a,b\ge 0,\label{eq:descendant vector potentials-4}
\end{align}
where we use the notation
$$
\mcF^{\alpha,a}_0:=(-1)^{a+1}q^\alpha_{-a-1},\quad \text{if $a<0$};\qquad q^\beta_b:=t^\beta_b-A^\beta\delta_{b,1}.
$$
The vector potential of the corresponding flat F-manifold is given by $\oF=(F^1,\ldots,F^L)$ where $F^\alpha:=\left.\mcF^{\alpha,0}_0\right|_{t^*_{\ge 1}=0}$, with the unit $A^\alpha\frac{\d}{\d t^\alpha}$.

\medskip

A collection of descendant vector potentials $\omcF^a_0$ is called {\it ancestor}, if $\ot_\orig=0$ and $\left.\frac{\d\mcF^{\alpha,a}_0}{\d t^\beta_b}\right|_{t^*_*=0}=0$ for all $a,b\ge 0$. For a given solution $\oF=(F^1,\ldots,F^L)$, $F^\alpha\in\mbC[[t^*]]$, to the oriented WDVV equations admitting a unit, there exists a unique sequence of ancestor vector potentials $\omcF^a_0$ such that $F^\alpha=\left.\mcF^{\alpha,0}_0\right|_{t^*_{\ge 1}=0}$.

\medskip

\begin{lemma}\label{lemma:vanishing in genus 0}
Consider an arbitrary collection of ancestor vector potentials $\omcF^a_0$. Then the coefficient of $t^{\alpha_1}_{d_1}\ldots t^{\alpha_n}_{d_n}$ in $\mcF^{\alpha,d}_0$ is zero if $d+\sum d_i\ge n-1$.
\end{lemma}
\begin{proof}
The statement if obviously true for $d=d_1=\ldots=d_n=0$. Then the statement is proved by the induction on $d+\sum d_i$ using equations~\eqref{eq:descendant vector potentials-3} and~\eqref{eq:descendant vector potentials-4}.
\end{proof}

\medskip

\begin{remark}\label{remark:whole collection from first vector}
It is easy to see that equations~\eqref{eq:descendant vector potentials-2} and~\eqref{eq:descendant vector potentials-4} imply that the whole collection of descendant vector potentials can be uniquely reconstructed from the $L$-tuple $(\mcF^{1,0}_0,\ldots,\mcF^{L,0}_0)$. 
\end{remark}

\medskip

\subsection{Group actions on descendant vector potentials}

\subsubsection{Linear changes of variables}

By~\cite[Section~2.1]{ABLR20}, there is a $\GL(\mbC^L)$-action on the space of descendant vector potentials of flat F-manifolds of dimension $L$ given by $\omcF^a_0\mapsto M.\omcF^a_0$, $M=(M^\alpha_\beta)\in\GL(\mbC^L)$, where
$$
M.\mcF^{\alpha,a}_0:=\left.M^\alpha_\mu\mcF^{\mu,a}_0\right|_{t^\beta_b\mapsto(M^{-1})^\beta_\gamma t^\gamma_b},
$$
and the unit $\oA=(A^1,\ldots,A^L)$ changes by the formula $\oA\mapsto M\oA$. 

\medskip

\subsubsection{Loop group actions}

Consider the following groups:
\begin{align*}
G_{L,+}:=&\left\{R(z)=\Id+\sum_{i\ge 1}R_i z^i\in\End(\mbC^{L})[[z]]\right\},\\
G_{L,-}:=&\left\{S(z)=\Id+\sum_{i\ge 1}S_i z^{-i}\in\End(\mbC^{L})[[z^{-1}]]\right\},
\end{align*}
and denote by $\mfg_{L,+}$ and $\mfg_{L,-}$ the corresponding Lie algebras.

\medskip

In~\cite{ABLR20} the authors constructed an action of the group $G_{L,+}$ on the space of ancestor vector potentials. Infinitesimally, it is given by the formula (see \cite[Proposition~2.9]{ABLR20})
\begin{align}
r(z).\mcF_0^{\alpha,a}=&\left.\frac{d}{d\theta}\left(e^{\theta r(z)}.\mcF_0^{\alpha,a}\right)\right|_{\theta=0}\notag
\\
=&\sum_{i\ge 1}(-1)^i(r_i)^\alpha_\mu\mcF_0^{\mu,a+i}+\sum_{i\ge 1,\,j\ge 0}(-1)^{i-j-1}(r_i)^\mu_\nu\frac{\d\mcF_0^{\alpha,a}}{\d q^\mu_j}\mcF_0^{\nu,i-j-1}
\label{eq:r-action for flat F-manifolds}\\
=&\sum_{i\ge 1}(-1)^i(r_i)^\alpha_\mu\mcF_0^{\mu,a+i}+\sum_{j,k\ge 0}(-1)^k(r_{j+k+1})^\mu_\nu\frac{\d\mcF_0^{\alpha,a}}{\d q^\mu_j}\mcF_0^{\nu,k}-\sum_{i\ge 1,\,\ell\ge 0}(r_i)^\mu_\nu\frac{\d\mcF_0^{\alpha,a}}{\d q^\mu_{i+\ell}}q^\nu_\ell,\notag
\end{align}
where $r(z)\in\mfg_{L,+}$ and $a \in \mathbb{Z}_{\ge 0}$.

\medskip

In~\cite{ABLR20} the authors also constructed an action of the group $G_{L,-}$ on the space of descendant vector potentials $\omcF^a_0$, $\mcF^{\alpha,a}_0\in\mcR_L^{\ot_\orig}[[t^*_{\ge 1}]]$. If $S(z)\in G_{L,-}$, then $S(z).\mcF_0^{\alpha,a}\in\mcR_L^{\ot_\orig-S_1\overline{A}}[[t^*_{\ge 1}]]$ and (see \cite[Proposition~2.10]{ABLR20})
\begin{gather*}
S(z).\mcF_0^{\alpha,a}=e^{-\widehat{s(z)}}\left(\mcF_0^{\alpha,a}+\sum_{i=1}^a(-1)^i(S_i)^\alpha_\mu\mcF_0^{\mu,a-i}+(-1)^{a+1}\sum_{i\ge a+1}(S_i)^\alpha_\mu q^\mu_{i-a-1}\right),
\end{gather*}
where $s(z):=\log S(z)$ and $\widehat{s(z)}:=\sum_{\substack{i\ge 1\\j\ge 0}}(s_i)^\alpha_\beta q^\beta_{i+j}\frac{\d}{\d q^\alpha_j}$. The corresponding formula for the infinitesimal action is
\begin{gather}\label{eq:s-action for flat F-manifolds}
s(z).\mcF_0^{\alpha,a}=-\sum_{i\ge 1,\,j\ge 0}(s_i)^\gamma_\beta q^\beta_{i+j}\frac{\d\mcF_0^{\alpha,a}}{\d q^\gamma_j}+\sum_{i=1}^a(-1)^i(s_i)^\alpha_\mu\mcF_0^{\mu,a-i}+(-1)^{a+1}\sum_{i\ge a+1}(s_i)^\alpha_\mu q^\mu_{i-a-1}.
\end{gather}

\medskip

\subsubsection{Shifts and the $G_{L,-}$-action}\label{subsubsection:shifts and lower action}

The $G_{L,-}$-action can be used to describe how the ancestor vector potentials change when we shift the variables, $t^\alpha\mapsto t^\alpha+\theta^\alpha$, in a vector potential $\oF=(F^1,\ldots,F^L)$.

\medskip

Let $\oF=(F^1,\ldots,F^L)$, $F^\alpha\in\mbC[[t^1,\ldots,t^L]]$, be a vector potential of a flat F-manifold and let $\omcF^a_0$ be the corresponding ancestor vector potentials. Define $L\times L$ matrices $\Omega_j(t^*)=(\Omega^\alpha_{j;\beta}(t^*))$, $j\ge 0$, by
$$
\Omega^\alpha_{j;\beta}(t^*):=\left.\frac{\d\mcF_0^{\alpha,j}}{\d t^\beta_0}\right|_{t^*_{\ge 1}=0}\in\mbC[[t^*]].
$$
Consider a family of vector potentials $\oF_{\otheta}$, depending on formal parameters $\theta^1,\ldots,\theta^L$, $\otheta=(\theta^1,\ldots,\theta^L)$, defined by
$$
F^\alpha_{\otheta}:=F^\alpha|_{t^\beta\mapsto t^\beta+\theta^\beta}\in\mbC[[\theta^*]][[t^*]].
$$
Denote by $\omcF^a_{0,\otheta}$ the corresponding family of ancestor vector potentials. By~\cite[Lemma~2.12]{ABLR20} we have
\begin{gather}\label{eq:ancestor vector potentials at shifted points}
\omcF^a_{0,\otheta}=\left(\Id+\sum_{j\ge 1}(-1)^j\Omega_{j-1}(\theta^*)z^{-j}\right)^{-1}.\omcF^a_0.
\end{gather}

\medskip

\subsubsection{The semisimplicity condition}\label{subsubsection:semisimplicity}

The above group actions preserve the semisimplicity condition of the corresponding flat F-manifold at $\ot_\orig$. Indeed, $\GL(\mbC^L)$ acts on $T_{\ot_\orig}\mbC^L$ by changing the basis, while the groups $G_{L,\pm}$ do not change the algebra structure on $T_{\ot_\orig}\mbC^L$. So the semisimplicity condition is preserved.  

\medskip

\subsubsection{Transitivity in the semisimple case}\label{subsubsection:general transitivity}

In~\cite{ABLR20} the authors proved that the group $G_{L,+}$ acts transitively on the space of ancestor vector potentials defining a fixed semisimple algebra structure on $T_0\mbC^L$. Their approach is constructive, and we recall it, because we will need it for the proof of Theorem~\ref{theorem:main}.

\medskip

We consider a solution $\oF=(F^1,\ldots,F^L)$, $F^\alpha\in\mbC[[t^*]]$, to the oriented WDVV equations, admitting a unit $A^\alpha\frac{\d}{\d t^\alpha}$ and that is semisimple at~$0\in\mbC^L$. \emph{Canonical coordinates} $u^1,\ldots,u^L$ are formal coordinates on $\mbC^L$ around $0$, $u^i=u^i(t^*)\in\mbC[[t^*]]$, such that the vector fields $\frac{\d}{\d u^i}$ are the idempotents for the algebra structure in $T\mbC^L$ around~$0$ defined by $\oF$. Canonical coordinates are defined uniquely up to permutations and shifts $u^i\mapsto u^i+a_i$, $a_i\in\mbC$.

\medskip

Choosing canonical coordinates, consider the matrix $\tPsi=\tPsi(t^*)$ defined by $\tPsi:=\left(\frac{\d u^i}{\d t^\alpha}\right)$. Consider the diagonal matrix $U:=\diag(u^1,\ldots,u^L)$. There exists a unique diagonal matrix $\tD$ whose entries are one-forms and a unique matrix $\tGamma=\tGamma(t^*)$ with vanishing diagonal part such that
\[
\tD+[\tGamma,dU]:=d\tPsi\cdot\tPsi^{-1}.
\]
Define a diagonal matrix $H=H(t^*)=\diag(H_1,\ldots,H_L)$, $H_i(0)\ne 0$, by $dH\cdot H^{-1}:=-\tD$. The functions $H_i\in\mbC[[t^*]]$ are defined uniquely up to the rescaling $H_i\mapsto\lambda_i H_i$ for $\lambda_i\in\mbC^*$. Define matrices $\Gamma=\Gamma(t^*)=(\gamma^i_j)$ and $\Psi=\Psi(t^*)$ by $\Gamma:=H\tGamma H^{-1}$ and $\Psi:=H\tPsi$. Note that the functions $\gamma^i_j$ can be expressed in terms of the functions $H_k$ as
$$
\gamma^i_j=H_j^{-1}\frac{\d H_i}{\d u^j},\quad i\ne j.
$$

\medskip

There exists a sequence of matrices $R_i=R_i(t^*)$, $i\ge 0$, $R_0=\Id$, satisfying the relations 
\begin{equation}\label{eq:Rmatrix recursion}
dR_{k-1}+R_{k-1}[\Gamma,dU]=[R_k,dU],\quad k\ge 1.
\end{equation}
In more details, the system of equations~\eqref{eq:Rmatrix recursion} is equivalent to the following system of recursive relations:
\begin{align}
&(R_{m+1})^i_j=(R_m)^i_i\gamma^i_j-\frac{\d(R_m)^i_j}{\d u^i},&& i\ne j, &&m\ge 0,\label{eq:R-relation,nondiagonal}\\
&d(R_{m+1})_i^i=-\sum_{j\ne i}(R_{m+1})^i_j\gamma_i^j(du^i-du^j),&& && m\ge 0,\label{eq:R-relation,diagonal}
\end{align}
which, given matrices $R_0,\ldots,R_m$, determine the matrix $R_{m+1}$ uniquely up to integration constants in the diagonal entries that can be arbitrary. 

\medskip

For $1\le\alpha\le L$, denote by $F^\alpha_\quadr$ the quadratic part of the formal power series $F^\alpha$. The functions $F^\alpha_\quadr$ are uniquely (up to linear polynomials in $t^*$) determined by the structure constants of the algebra structure on $T_0\mbC^L$, they satisfy the oriented WDVV equations and admit the same unit $A^\alpha\frac{\d}{\d t^\alpha}$. By~\cite[proof of Theorem~2.13]{ABLR20}, the sequence of ancestor vector potentials corresponding to the vector potential $\oF$ can be obtained from the sequence of ancestor vector potentials corresponding to the vector potential $\oF_\quadr:=(F^1_\quadr,\ldots,F^M_\quadr)$ by the action of the element $\mcR(z)\in G_{L,+}$ given by
$$
\mcR(z)=\Psi^{-1}(0)R^{-1}(-z,0)\Psi(0),
$$
where $R(z,t^*):=\Id+\sum_{i\ge 1} R_i(t^*)z^i$.

\medskip

\subsection{The Givental theory as a reduction of the general case}\label{subsection:classical Givental theory}

Let $L=N$.

\medskip

\subsubsection{Closed descendant potentials}\label{subsubsection:closed descendant potentials}

We will say that a collection of descendant vector potentials is of {\it closed type}, if there exists a constant symmetric nondegenerate matrix $\eta=(\eta_{\alpha\beta})$ such that
$$
\frac{\d(\eta_{\alpha\mu}\mcF^{\mu,a}_0)}{\d t^\beta_b}=\frac{\d(\eta_{\beta\mu}\mcF^{\mu,b}_0)}{\d t^\alpha_a},\quad 1\le\alpha,\beta\le N,\quad a,b\ge 0.
$$
In this case we define a function $\mcF_0^c\in\mcR_N^{\ot_\orig}[[t^*_{\ge 1}]]$ by
$$
\mcF_0^c:=\frac{1}{2}\sum_{a\ge 0}q^\alpha_a\eta_{\alpha\mu}\mcF^{\mu,a}_0,
$$
which satisfies
$$
\frac{\d\mcF_0^c}{\d t^\beta_b}=\eta_{\beta\mu}\mcF^{\mu,b}_0,\qquad\sum_{a\ge 0}q^\alpha_a\frac{\d\mcF^c_0}{\d q^\alpha_a}=2\mcF^c_0.
$$
We call $\mcF^c_0$ a {\it closed descendant potential}. The function $F^c:=\left.\mcF^c_0\right|_{t^*_{\ge 1}=0}$ then satisfies the closed WDVV equations and admits the unit $A^\alpha\frac{\d}{\d t^\alpha}$. If $\mcF^c_0$ comes from a sequence of ancestor vector potentials, then $\mcF^c_0$ is called a \emph{closed ancestor potential}. 

\medskip

The space of closed descendant potentials can be alternatively described as the space of solutions $\mcF^c_0\in\mcR^{\ot_\orig}_N[[t^*_{\ge 1}]]$ to the following equations: 
\begin{align}
\sum_{n\ge 0}q^\gamma_{n+1}\frac{\d\mcF^c_0}{\d q^\gamma_n}=&-\frac{1}{2}\eta_{\alpha\beta}t^\alpha_0 t^\beta_0,\label{eq:closed string-0}\\
\sum_{b\ge 0}q^\beta_b\frac{\d\mcF^c_0}{\d q^\beta_b}=&2\mcF_0^c,\label{eq:closed dilaton-0}\\
\frac{\d^3\mcF_0^c}{\d q^\alpha_{a+1}\d q^\beta_b\d q^\gamma_c}=&\frac{\d^2\mcF_0^c}{\d q^\alpha_a\d q^\mu_0}\eta^{\mu\nu}\frac{\d^3\mcF_0^c}{\d q^\nu_0\d q^\beta_b\d q^\gamma_c},&& 1\le\alpha,\beta,\gamma\le N,\quad a,b,c\ge 0,\label{eq:closed TRR-0}\\
\frac{\d^2\mcF_0^c}{\d q^\alpha_{a+1}\d q^\beta_b}+\frac{\d^2\mcF_0^c}{\d q^\alpha_a\d q^\beta_{b+1}}=&\frac{\d^2\mcF_0^c}{\d q^\alpha_a\d q^\mu_0}\eta^{\mu\nu}\frac{\d^2\mcF_0^c}{\d q^\nu_0\d q^\beta_b},&& 1\le\alpha,\beta\le N,\quad a,b\ge 0.\notag
\end{align}
This matches the Givental approach to the descendant potentials from~\cite{Giv04}. Equations~\eqref{eq:closed string-0}, \eqref{eq:closed dilaton-0}, and~\eqref{eq:closed TRR-0} are called, respectively, the \emph{closed string}, \emph{dilaton}, and \emph{topological recursion relations in genus $0$}. The subspace of ancestor potentials is specified by the conditions $\ot_\orig=0$ and $\left.\frac{\d^2\mcF^c_0}{\d t^\alpha_a\d t^\beta_b}\right|_{t^*_*=0}=0$ for all $1\le\alpha,\beta\le N$ and $a,b\ge 0$. From Lemma~\ref{lemma:vanishing in genus 0} it follows that the coefficient of $t^{\alpha_1}_{d_1}\ldots t^{\alpha_n}_{d_n}$ in an arbitrary closed ancestor potential $\mcF^c_0$ is zero if $\sum d_i\ge n-2$. 

\medskip

\begin{example}
A basic example of a closed ancestor potential is the function 
$$
\mcF^c_0=\mcF^\KW_0=\frac{t_0^3}{6}+\frac{t_0^3 t_1}{6}+\left(\frac{t_0^3t_1^2}{6}+\frac{t_0^4t_2}{24}\right)+O\left((t_*)^6\right),
$$
the corresponding metric is $\eta=1$ and the unit is $e=\frac{\d}{\d t^1}$. There is a simple formula for the coefficients of $\mcF^\KW_0$:
$$
\left.\frac{\d^n\mcF^\KW_0}{\d t_{d_1}\ldots\d t_{d_n}}\right|_{t_*=0}=
\begin{cases}
\frac{(n-3)!}{\prod d_i!},&\text{if $n\ge 3$ and $\sum d_i=n-3$},\\
0,&\text{otherwise}.
\end{cases}
$$
\end{example}

\medskip

\subsubsection{Group actions}\label{subsubsection:closed actions in genus 0}

The $\GL(\mbC^N)$-action on the space of descendant vector potentials preserves the subspace of descendant vector potentials of closed type and acts on the metric $\eta$ and on the closed descendant potential $\mcF_0^c$ by
$$
M.\eta_{\alpha\beta}=(M^{-1})_\alpha^\mu\eta_{\mu\nu}(M^{-1})^\nu_\beta,\qquad M.\mcF^c_0=\left.\mcF^c_0\right|_{t^\beta_b\mapsto(M^{-1})^\beta_\gamma t^\gamma_b}, 
$$
for $M\in\GL(\mbC^N)$. Note that if $M=\exp(m)$, for some matrix $m=(m^\alpha_\beta)$, then
$$
M.\mcF^c_0=\exp(-\widehat{m})\mcF^c_0,
$$
where
\begin{gather}\label{eq:definition of mhat}
\widehat{m}:=\sum_{d\ge 0}m^\beta_\alpha t^\alpha_d\frac{\d}{\d t^\beta_d}.
\end{gather}

\medskip

For a fixed matrix $\eta$ consider the following subgroups of the groups $G_{N,\pm}$:
\begin{gather*}
G_{N,\pm}^c:=\left\{\left.M(z)=\Id+\sum\nolimits_{i\ge 1}M_i z^{\pm i}\in \End(\mbC^N)[[z^{\pm 1}]]\right|M^T(-z)\eta M(z)=\eta\right\}\subset G_{N,\pm}.
\end{gather*}
The corresponding Lie algebras are given by
\begin{gather*}
\mfg_{N,\pm}^c=\left\{\left.m(z)=\sum\nolimits_{i\ge 1}m_i z^{\pm i}\in z^{\pm 1}\End(\mbC^N)[[z^{\pm 1}]]\right|\eta^{\alpha\mu}(m_i)^\nu_\mu\eta_{\nu\beta}=(-1)^{i+1}(m_i)^\alpha_\beta\right\}.
\end{gather*}
Let us introduce the following differential operators:
\begin{align*}
\widehat{r(z)}^c:=&-\sum_{i\ge 1,\,j\ge 0}(r_i)^\mu_\nu q^\nu_j\frac{\d}{\d q^\mu_{i+j}}+\frac{\eps^2}{2}\sum_{i,j\ge 0}(-1)^j(r_{i+j+1})^\alpha_\mu\eta^{\mu\beta}\frac{\d^2}{\d q^\alpha_i\d q^\beta_j},&&r(z)\in\mfg_{N,+}^c,\\
\widehat{s(z)}^c:=&-\sum_{i\ge 1,\,j\ge 0}(s_i)^\alpha_\beta q^\beta_{i+j}\frac{\d}{\d q^\alpha_j}+\frac{\eps^{-2}}{2}\sum_{i,j\ge 0}(-1)^{j+1}(s_{i+j+1})_{\alpha}^\mu\eta_{\mu\beta}q^\alpha_i q^\beta_j,&&s(z)\in\mfg_{N,-}^c.
\end{align*}
The fact that the group $G^c_{N,+}$ ($G^c_{N,-}$) preserves the space of ancestor (descendant) vector potentials of closed type is well known from~\cite{Giv04} (see also~\cite{L05}). However, for completeness, let us present a quick derivation of this result.

\medskip

\begin{lemma}\label{lemma:Givental action}
{\ }
\begin{enumerate}[ 1.]
\item The group $G^c_{N,+}$ ($G^c_{N,-}$) preserves the space of ancestor (descendant) vector potentials of closed type. We, therefore, denote by $r(z).\mcF^c_0$ ($s(z).\mcF^c_0$) the resulting infinitesimal action of $\mfg^c_{N,+}$ ($\mfg^c_{N,-}$) on a closed ancestor (descendant) potential $\mcF^c_0$. 

\smallskip

\item The corresponding action on the space of closed ancestor (descendant) potentials is given~by
\begin{align}
\frac{\widehat{r(z)}^c\exp(\eps^{-2}\mcF^c_0)}{\exp(\eps^{-2}\mcF^c_0)}=&\eps^{-2}r(z).\mcF^c_0+O(\eps^0),&& r(z)\in\mfg_{N,+}^c,\label{eq:closed r-action,difop}\\
\frac{\widehat{s(z)}^c\exp(\eps^{-2}\mcF^c_0)}{\exp(\eps^{-2}\mcF^c_0)}=&\eps^{-2}s(z).\mcF^c_0+O(\eps^0),&& s(z)\in\mfg_{N,-}^c.\label{eq:closed s-action,difop}
\end{align}
\end{enumerate}
\end{lemma}
\begin{proof}
For the group $G^c_{N,+}$ we compute
\begin{align*}
\eta^{\alpha\xi}\frac{\d}{\d t^\xi_a}\Coef_{\eps^{-2}}&\left(\frac{\widehat{r(z)}^c\exp(\eps^{-2}\mcF^c_0)}{\exp(\eps^{-2}\mcF^c_0)}\right)=\\
=&\eta^{\alpha\xi}\frac{\d}{\d t^\xi_a}\left(-\sum_{i\ge 1,\,j\ge 0}(r_i)^\mu_\nu q^\nu_j\frac{\d\mcF_0^c}{\d t^\mu_{i+j}}+\frac{1}{2}\sum_{i,j\ge 0}(-1)^j(r_{i+j+1})^\gamma_\mu\eta^{\mu\beta}\frac{\d\mcF_0^c}{\d t^\gamma_i}\frac{\d\mcF_0^c}{\d t^\beta_j}\right)=\\
=&-\sum_{i\ge 1}\eta^{\alpha\nu}(r_i)^\mu_\nu \frac{\d\mcF_0^c}{\d t^\mu_{i+a}}-\sum_{i\ge 1,\,j\ge 0}(r_i)^\mu_\nu q^\nu_j\frac{\d\mcF_0^{\alpha,a}}{\d t^\mu_{i+j}}+\sum_{i,j\ge 0}(-1)^j(r_{i+j+1})^\gamma_\mu\frac{\d\mcF_0^{\alpha,a}}{\d t^\gamma_i}\mcF_0^{\mu,j}=\\
=&\sum_{i\ge 1}(-1)^i(r_i)^\alpha_\mu\mcF_0^{\mu,i+a}-\sum_{i\ge 1,\,j\ge 0}(r_i)^\mu_\nu q^\nu_j\frac{\d\mcF_0^{\alpha,a}}{\d t^\mu_{i+j}}+\sum_{i,j\ge 0}(-1)^j(r_{i+j+1})^\gamma_\mu\frac{\d\mcF_0^{\alpha,a}}{\d t^\gamma_i}\mcF_0^{\mu,j}=\\
=&r(z).\mcF_0^{\alpha,a},
\end{align*}
and also notice that $\Coef_{\eps^{-2}}\left(\frac{\widehat{r(z)}^c\exp(\eps^{-2}\mcF^c_0)}{\exp(\eps^{-2}\mcF^c_0)}\right)$ is an eigenvector for the operator $\sum_{a\ge 0}q^\alpha_a\frac{\d}{\d q^\alpha_a}$ with eigenvalue $2$. The proof for the group $G^c_{N,-}$ is similar.
\end{proof}

\medskip

\subsubsection{A transitivity statement for the $G^c_{N,+}$-action}

\medskip

Givental proved that the group $G^c_{N,+}$ acts transitively on the space of closed ancestor potentials defining a fixed semisimple algebra structure on $T_0\mbC^N$. Let us see how to deduce this result from the analogous result about the transitivity of the $G_{N,+}$-action described in Section~\ref{subsubsection:general transitivity}.

\medskip

Indeed, we consider a closed ancestor potential $\mcF^c_0$, the associated solution $F^c$ to the closed WDVV equations, the associated sequence of ancestor vector potentials, and follow the constructions from Section~\ref{subsubsection:general transitivity}. A standard result in the theory of Dubrovin--Frobenius manifolds~\cite{Dub96} says that the symmetric bilinear form $\eta=\frac{1}{2}\eta_{\alpha\beta}d t^\alpha d t^\beta$ becomes diagonal in the canonical coordinates, $\eta=\sum_{i=1}^N g_i(du^i)^2$, and one can choose the functions $H_i$ such that $H_i^2=g_i$. In this case the matrices $R_i$ can be chosen~\cite{Giv01Sem} in such a way that they satisfy the additional orthogonality condition
$$
R(-z,t^*)^TR(z,t^*)=\Id.
$$
Then $\mcR(z)=\Psi^{-1}(0)R^{-1}(-z,0)\Psi(0)\in G^c_{N,+}$. If we denote by $F^c_{\cub}$ the cubic part of the formal power series $F^c$, then the closed ancestor potential $\mcF^c_0$ is obtained from the closed ancestor potential associated to $F^c_\cub$ by the action of the element $\mcR(z)\in G^c_{N,+}$.

\medskip

\subsection{Open descendant potentials}

Let $L=N+1$.

\medskip

\subsubsection{Definition}\label{subsubsection:definition of open descendant}

We say that a sequence of descendant vector potentials $\omcF^a_0=(\mcF^{1,a}_0,\ldots,\mcF^{N+1,a}_0)$ is of {\it open-closed type}, if the functions $\mcF^{1,a}_0,\ldots,\mcF^{N,a}_0$ do not depend on $t^{N+1}_*$ and the sequence of $N$-tuples of functions $(\mcF^{1,a}_0,\ldots,\mcF^{N,a}_0)$ is a sequence of descendant vector potentials of closed type, with some matrix $\eta$ and closed descendant potential $\mcF^c_0$. The sequence of descendant vector potentials $(\mcF^{1,a}_0,\ldots,\mcF^{N,a}_0)$ is called the {\it closed part} of the vector potentials $\omcF^a_0$. We will denote $\mcF^{o,a}_0:=\mcF^{N+1,a}_0$ and $\mcF^o_0:=\mcF^{o,0}_0$.

\medskip

By Remark~\ref{remark:whole collection from first vector}, the whole collection of descendant vector potentials $\omcF^a_0$ of open-closed type is uniquely determined by the pair $(\mcF^c_0,\mcF^o_0)$. The function $\mcF^o_0$ is called the {\it open descendant potential}. Note that equation~\eqref{eq:descendant vector potentials-3} written in terms of the functions $\mcF^c_0$ and $\mcF^o_0$ looks as follows:
$$
\frac{\d^2\mcF_0^o}{\d t^\alpha_{a+1}\d t^\beta_b}=\frac{\d^2\mcF_0^c}{\d t^\alpha_a\d t^\mu_0}\eta^{\mu \nu}\frac{\d^2\mcF_0^o}{\d t^\nu_0\d t^\beta_b}+\frac{\d\mcF^o_0}{\d t^\alpha_a}\frac{\d^2\mcF^o_0}{\d t_0^{N+1}\d t^\beta_b}, \quad 1\le\alpha,\beta\le N+1, \quad a,b\ge 0,
$$
and one immediately recognizes here the open TRR-$0$ relations~\eqref{eq:open TRR0}. The open descendant potential $\mcF^o_0$ is called \emph{ancestor} if the associated collection of descendant vector potentials is ancestor. From Lemma~\ref{lemma:vanishing in genus 0} it follows that the coefficient of $t^{\alpha_1}_{d_1}\ldots t^{\alpha_n}_{d_n}$ in an arbitrary open ancestor potential $\mcF^o_0$ is zero if $\sum d_i\ge n-1$.

\medskip

\begin{example}
A basic example of a pair of closed and open ancestor potentials is the pair $\left(\mcF^\KW_0(t^1_*),\mcF^\PST_0(t^1_*,t^2_*)\right)$, the metric is $\eta=1$ and the unit is $\frac{\d}{\d t^1}$ \cite{PST14} (see also~\cite{BB19}). The first few terms of the formal power series $\mcF^\PST_0=\mcF^\PST_0(t_*,s_*)$ are given by
$$
\mcF^\PST_0=t_0 s_0+\left(t_0 t_1 s_0+\frac{s_0^3}{6}\right)+\left(\frac{t_0s_0^2s_1}{2}+\frac{t_1s_0^3}{3}+\frac{t_0^2t_2s_0}{2}+t_0t_1^2s_0+t_0^2t_1s_1+\frac{t_0^3s_2}{6}\right)+\ldots,
$$
where the dots contain the monomials $\prod_{i=1}^l t_{d_i}\prod_{j=1}^k s_{m_j}$ with $l+k\ge 5$. An explicit formula for the coefficients of $\mcF^{\PST}_0$ was found in~\cite{PST14}:
$$
\left.\frac{\d^{l+k}\mcF^\PST_0}{\d t_{d_1}\ldots\d t_{d_l}\d s_0^k}\right|_{t_*=s_*=0}=
\frac{(2\sum d_i-l+1)!}{\prod(2d_i-1)!!},\quad\text{if $2l+k\ge 3$, $2\sum d_i=2l+k-3$, and $d_i\ge 1$}.
$$
All the remaining coefficients of $\mcF^\PST_0$ can be found using the open string equation and the relation~\eqref{eq:descendants of s in PST}.
\end{example} 

\medskip

For an arbitrary pair of closed and open descendant potentials $(\mcF^c_0,\mcF^o_0)$ the functions $F^c:=\left.\mcF^c_0\right|_{t^*_{\ge 1}=0}$ and $F^o:=\left.\mcF^o_0\right|_{t^*_{\ge 1}=0}$ satisfy the closed and open WDVV equations and admit the unit $A^\alpha\frac{\d}{\d t^\alpha}$. Conversely, for arbitrary solutions $F^c,F^o\in\mcR_{N+1}^{\ot_\orig}$ to the closed and open WDVV equations admitting a unit, there exists a pair of closed and open descendant potentials $(\mcF^c_0,\mcF^o_0)$ such that $F^c=\left.\mcF^c_0\right|_{t^*_{\ge 1}=0}$ and $F^o=\left.\mcF^o_0\right|_{t^*_{\ge 1}=0}$. This is proved in~\cite[Lemma~3.2]{BB19} under additional homogeneity assumptions for $F^c$ and $F^o$, because the authors there also require that~$\mcF^c_0$ and~$\mcF^o_0$ satisfy additional homogeneity conditions. However, the same approach works in the general, nonhomogeneous, case that we consider here.  

\medskip

\subsubsection{Group actions}\label{subsubsection:group actions}

Consider the $\GL(\mbC^{N+1})$-action on the space of descendant vector potentials. Clearly, the subgroup $\tGL^o(\mbC^{N+1})\subset\GL(\mbC^{N+1})$ given by
$$
\tGL^o(\mbC^{N+1}):=\left\{M\in\GL(\mbC^{N+1})\left|M^{\le N}_{N+1}=0\right.\right\}
$$
preserves the space of descendant vector potentials of open-closed type. For an $(N+1)\times(N+1)$ matrix $M$ we denote by $\pi(M)$ the $N\times N$ matrix obtained by deleting the last row and the last column in $M$. For any $M\in\tGL^o(\mbC^{N+1})$ we have the following formula:
\begin{align*}
&M.(\mcF^c_0,\mcF^o_0)=\left(\pi(M).\mcF^c_0,M.\mcF^o_0\right),\quad\text{where}\\
&M.\mcF^o_0=\left.\left(M^{N+1}_{N+1}\mcF^o_0+M^{N+1}_\alpha\eta^{\alpha\mu}\frac{\d\mcF^c_0}{\d t^\mu_0}\right)\right|_{t^\beta_b\mapsto(M^{-1})^\beta_\gamma t^\gamma_b}.
\end{align*}

\medskip

For our future construction of open total descendant potentials, it is useful to consider in details the action of a certain subgroup of $\tGL^o(\mbC^{N+1})$. Define a subgroup $\GL^o(\mbC^{N+1})\subset\tGL^o(\mbC^{N+1})$ by
$$
\GL^o(\mbC^{N+1}):=\left\{\left.M\in\tGL^o(\mbC^{N+1})\right|M^{N+1}_{N+1}=1\right\}.
$$
Embedding the groups~$\GL(\mbC^N)$ and $\mbC^N$ into $\GL^o(\mbC^{N+1})$ by
\begin{align*}
&\GL(\mbC^N)\ni M\mapsto
M^\ext:=\scriptsize{ 
\left(
\begin{array}{c|c}
M & 
         {\begin{array}{c}
          0 \\ \vdots \\ 0
         \end{array}}
\\
\hline
0 \ldots 0 & 1
\end{array}
\right)}\in\GL^o(\mbC^{N+1}),\\
&\mbC^N\ni \ov=(v_1,\ldots,v_N)\mapsto B_{\ov}:=
\scriptsize{ 
\left(
\begin{array}{c|c}
\Id & 
         {\begin{array}{c}
          0 \\ \vdots \\ 0
         \end{array}}
\\
\hline
v_1 \ldots v_N & 1
\end{array}
\right)}\in\GL^o(\mbC^{N+1}),
\end{align*}
we have the decomposition
$$
\GL^o(\mbC^{N+1})=\mbC^N\rtimes\GL(\mbC^N).
$$
For $M\in\GL(\mbC^N)$ we have the formula
\begin{gather}\label{eq:hatm-action}
\log\left(\exp(-\widehat{m})\exp(\eps^{-2}\mcF^c_0+\eps^{-1}\mcF^o_0)\right)=\eps^{-2}M.\mcF^c_0+\eps^{-1}M^\ext.\mcF^o_0+O(\eps^0),
\end{gather}
where $\exp(m)=M$ and we recall that $\widehat{m}$ is defined by~\eqref{eq:definition of mhat}. For $\ov\in\mbC^N$ we have the formula
\begin{gather}\label{eq:hatov-action}
\log\left(\exp(-\widehat{\ov})\exp(\eps^{-2}\mcF^c_0+\eps^{-1}\mcF^o_0)\right)=\eps^{-2}\mcF^c_0+\eps^{-1}B_\ov.\mcF^o_0+O(\eps^0),
\end{gather}
where 
$$
\widehat{\ov}:=\sum_{d\ge 0}v_\alpha t^\alpha_d\frac{\d}{\d t^{N+1}_d}-\eps v_\alpha\eta^{\alpha\beta}\frac{\d}{\d t^\beta_0}.
$$

\medskip

Let us now study transformations from the groups $G_{N+1,+}$ ($G_{N+1,-}$) that preserve the space of ancestor (descendant) vector potentials of open-closed type. Consider such a collection of ancestor (descendant) vector potentials. Looking at formula~\eqref{eq:r-action for flat F-manifolds} (formula~\eqref{eq:s-action for flat F-manifolds}), one can see that if $(r_i)^\alpha_{N+1}=0$ ($(s_i)^\alpha_{N+1}=0$) for $1\le\alpha\le N$, then the functions $\mcF^{o,a}_0$ do not contribute to the deformation of the functions $\mcF^{\alpha,a}_0$ with $1\le\alpha\le N$. If we further require that $\pi(r(z))\in \mfg^c_{N,+}$ ($\pi(s(z))\in \mfg^c_{N,-}$), then our action restricts to the classical Givental action on the closed part of our descendant vector potentials.

\medskip

We see that the subgroup $\tG^o_{N+1,+}\subset G_{N+1,+}$ defined by
\begin{gather*}
\tG_{N+1,+}^o:=\left\{R(z)=\Id+\sum\nolimits_{i\ge 1}R_i z^i\in \End(\mbC^{N+1})[[z]]\left|\begin{minipage}{2.9cm}$\pi(R(z))\in G^c_{N,+}$\\$(R_i)^{\le N}_{N+1}=0$\end{minipage}\right.\right\}
\end{gather*}
preserves the space of ancestor vector potentials of open-closed type, and the subgroup $G^o_{N+1,-}\subset G_{N+1,-}$ defined by
\begin{gather*}
G_{N+1,-}^o:=\left\{S(z)=\Id+\sum\nolimits_{i\ge 1}S_i z^{-i}\in \End(\mbC^{N+1})[[z^{-1}]]\left|\begin{minipage}{2.9cm}$\pi(S(z))\in G^c_{N,-}$\\$(S_i)^{\le N}_{N+1}=0$\end{minipage}\right.\right\}
\end{gather*}
preserves the space of descendant vector potentials of open-closed type.

\medskip

Let us now discuss the possibility to express the action of the group $\tG^o_{N+1,+}$ ($G^o_{N+1,-}$) on pairs of ancestor (descendant) potentials $(\mcF^c_0,\mcF^o_0)$ in a form similar to~\eqref{eq:closed r-action,difop} and~\eqref{eq:closed s-action,difop}.

\medskip

Let us first discuss the group $\tG^o_{N+1,+}$, whose Lie algebra we denote by $\tmfg^o_{N+1,+}$. Given $r(z)\in\tmfg^o_{N+1,+}$, its action on $\mcF^c_0$ is given by equation~\eqref{eq:closed r-action,difop}, while the action of $r(z)$ on the function~$\mcF^o_0$ is given by
\begin{gather}\label{eq:preliminary upper open deformation}
r(z).\mcF^o_0=\sum_{i\ge 1}(-1)^i(r_i)^{N+1}_\mu\mcF_0^{\mu,i}+\sum_{j,k\ge 0}(-1)^k(r_{j+k+1})^\mu_\nu\frac{\d\mcF_0^o}{\d q^\mu_j}\mcF_0^{\nu,k}-\sum_{i\ge 1,\,l\ge 0}(r_i)^\mu_\nu\frac{\d\mcF^o_0}{\d q^\mu_{i+l}}q^\nu_l.
\end{gather}
We see that the functions $\mcF^{o,i}_0$ with $i\ge 1$ appear in general on the right-hand side of this formula. However, note that if we require that $(r_i)^{N+1}_{N+1}=0$, then the functions $\mcF^{o,i}_0$ with $i\ge 1$ do not appear on the right-hand side of~\eqref{eq:preliminary upper open deformation}, and we get
\begin{gather}\label{eq:upper open deformation}
r(z).\mcF^o_0=\sum_{i\ge 1}(-1)^i(r_i)^{N+1,\nu}\frac{\d\mcF_0^c}{\d q^\nu_i}+\sum_{j,k\ge 0}(-1)^k(r_{j+k+1})^{\mu\gamma}\frac{\d\mcF_0^o}{\d q^\mu_j}\frac{\d\mcF_0^c}{\d q^\gamma_k}-\sum_{i\ge 1,\,l\ge 0}(r_i)^\mu_\nu q^\nu_l\frac{\d\mcF_0^o}{\d q^\mu_{i+l}},
\end{gather}
where we use the following nonstandard (!) notation:
\begin{gather*}
(r_i)^{\alpha\beta}:=
\begin{cases}
(r_i)^\alpha_\mu\eta^{\mu\beta},&\text{if $\beta\le N$},\\
(-1)^{i+1}(r_i)^{N+1,\alpha},&\text{if $\beta=N+1$ and $\alpha\le N$},\\
0,&\text{if $\alpha=\beta=N+1$}.
\end{cases}
\end{gather*}

\medskip

We therefore introduce the following subgroup of $\tG^o_{N+1,+}$:
\begin{gather*}
G_{N+1,+}^o:=\left\{R(z)=\Id+\sum\nolimits_{i\ge 1}R_i z^i\in \End(\mbC^{N+1})[[z]]\left|\begin{minipage}{2.9cm}$\pi(R(z))\in G^c_{N,+}$\\$(R_i)^*_{N+1}=0$\end{minipage}\right.\right\}\subset\tG^o_{N+1,+},
\end{gather*}
and denote by $\mfg_{N+1,+}^o$ the corresponding Lie algebra. Remarkably, if for any $r(z)\in\mfg^o_{N+1,+}$ we introduce the differential operator
\begin{gather*}
\widehat{r(z)}^o:=-\sum_{i\ge 1,\,j\ge 0}(r_i)^\mu_\nu q^\nu_j\frac{\d}{\d q^\mu_{i+j}}+\eps\sum_{i\ge 1}(-1)^i(r_i)^{N+1,\nu}\frac{\d}{\d q^\nu_i}+\frac{\eps^2}{2}\sum_{i,j\ge 0}(-1)^j(r_{i+j+1})^{\alpha\beta}\frac{\d^2}{\d q^\alpha_i\d q^\beta_j},
\end{gather*}
then formulas~\eqref{eq:closed r-action,difop} and~\eqref{eq:upper open deformation} are combined in the following way:
\begin{gather}\label{eq:open infinitesimal r-action}
\frac{\widehat{r(z)}^o \exp(\eps^{-2}\mcF^c_0+\eps^{-1}\mcF^o_0)}{\exp(\eps^{-2}\mcF^c_0+\eps^{-1}\mcF^o_0)}=\eps^{-2}\pi(r(z)).\mcF^c_0+\eps^{-1}r(z).\mcF^o_0+O(\eps^0).
\end{gather}

\medskip

Consider now the group $G^o_{N+1,-}$, whose Lie algebra we denote by $\mfg^o_{N+1,-}$. Given $s(z)\in\mfg^o_{N+1,-}$, its action on $\mcF^c_0$ is given by equation~\eqref{eq:closed s-action,difop}, while the action of $s(z)$ on the function~$\mcF^o_0$ is given~by
\begin{gather}\label{eq:lower open deformation}
s(z).\mcF^o_0=-\sum_{i\ge 1,\,j\ge 0}(s_i)^\gamma_\beta q^\beta_{i+j}\frac{\d\mcF_0^o}{\d q^\gamma_j}-\sum_{i\ge 1}(s_i)^{N+1}_\mu q^\mu_{i-1}.
\end{gather}
Again, if for any $s(z)\in\mfg^o_{N+1,-}$ we introduce the differential operator
\begin{gather*}
\widehat{s(z)}^o:=-\sum_{i\ge 1,\,j\ge 0}(s_i)^\alpha_\beta q^\beta_{i+j}\frac{\d}{\d q^\alpha_j}-\eps^{-1}\sum_{i\ge 1}(s_i)^{N+1}_\mu q^{\mu}_{i-1}+\frac{\eps^{-2}}{2}\sum_{i,j\ge 0}(-1)^{j+1}(s_{i+j+1})_{\alpha}^\mu\eta_{\mu\beta}q^\alpha_i q^\beta_j,
\end{gather*}
then formulas~\eqref{eq:closed s-action,difop} and~\eqref{eq:lower open deformation} are combined as follows:
\begin{gather}\label{eq:open infinitesimal s-action}
\frac{\widehat{s(z)}^o \exp(\eps^{-2}\mcF^c_0+\eps^{-1}\mcF^o_0)}{\exp(\eps^{-2}\mcF^c_0+\eps^{-1}\mcF^o_0)}=\eps^{-2}\pi(s(z)).\mcF^c_0+\eps^{-1}s(z).\mcF^o_0+O(\eps^0).
\end{gather}

\medskip

In the following proposition, we summarize what we have just obtained.

\medskip

\begin{proposition}\label{lemma:open Givental action}
The group $G^o_{N+1,+}$ ($G^o_{N+1,-}$) preserves the space of ancestor (descendant) vector potentials of open-closed type. Moreover, the corresponding action on the space of pairs of closed and open ancestor (descendant) potentials $(\mcF^c_0,\mcF^o_0)$ is given by equations~\eqref{eq:open infinitesimal r-action} and~\eqref{eq:open infinitesimal s-action}.
\end{proposition}

\medskip

Let us present the following useful technical result about the structure of the operators $\exp\left(\widehat{s(z)}^o\right)$ for $s(z)\in\mfg^o_{N+1,-}$.

\medskip

\begin{lemma}\label{lemma:exponent of open s}
For any $s(z)\in\mfg^o_{N+1,-}$ we have
\begin{align*}
&\exp\left(\widehat{s(z)}^o\right)=\\
=&\exp\left(\eps^{-1}\sum_{i\ge 1}\left((S^{-1})_i\right)^{N+1}_\alpha q^\alpha_{i-1}+\frac{\eps^{-2}}{2}\sum_{a,b\ge 0}q^\alpha_a q^\beta_b\Coef_{z_1^{-a} z_2^{-b}}\left(\frac{S(z_1)_\alpha^\mu \eta_{\mu\nu} S(z_2)^\nu_\beta-\eta_{\alpha\beta}}{z_1^{-1}+z_2^{-1}}\right)\right)\times\\
&\times\exp\left(-\sum_{i\ge 1,\,j\ge 0}(s_i)^\alpha_\beta q^\beta_{i+j}\frac{\d}{\d q^\alpha_j}\right),
\end{align*}
where $S(z):=\exp(s(z))$.
\end{lemma}
\begin{proof}
We apply part (a) of Lemma~\ref{lemma:BCH-special} with
\begin{gather*}
X:=-\sum_{i\ge 1,\,j\ge 0}(s_i)^\alpha_\beta q^\beta_{i+j}\frac{\d}{\d q^\alpha_j}
\end{gather*}
and $Y:=Y_1+Y_2$ where
$$
Y_1:=-\eps^{-1}\sum_{i\ge 1}(s_i)^{N+1}_\mu q^{\mu}_{i-1},\qquad Y_2:=\frac{\eps^{-2}}{2}\sum_{i,j\ge 0}(-1)^{j+1}(s_{i+j+1})_{\alpha}^\mu\eta_{\mu\beta}q^\alpha_i q^\beta_j.
$$
Then $\widehat{s(z)}^o=X+Y$ and moreover it is well known that 
$$
\sum_{n\ge 0}\frac{1}{(n+1)!}\ad^n_X Y_2=\frac{\eps^{-2}}{2}\sum_{a,b\ge 0}q^\alpha_a q^\beta_b\Coef_{z_1^{-a} z_2^{-b}}\left(\frac{S(z_1)_\alpha^\mu \eta_{\mu\nu} S(z_2)^\nu_\beta-\eta_{\alpha\beta}}{z_1^{-1}+z_2^{-1}}\right).
$$
The fact that
$$
\sum_{n\ge 0}\frac{1}{(n+1)!}\ad^n_X Y_1=\eps^{-1}\sum_{i\ge 1}\left((S^{-1})_i\right)^{N+1}_\alpha q^\alpha_{i-1}
$$
is proved by a simple direct computation.
\end{proof}

\medskip

\subsubsection{A transitivity statement for the $G^o_{N+1,+}$-action}

\begin{proposition}\label{proposition:open transitivity}
The group $G^o_{N+1,+}$ acts transitively on the space of pairs of closed and open ancestor potentials $(\mcF^c_0,\mcF^o_0)$ defining a fixed semisimple algebra structure on $T_0\mbC^{N+1}$.
\end{proposition}
\begin{proof}
We consider an arbitrary pair of ancestor potentials $(\mcF^c_0,\mcF^o_0)$ defining a semisimple algebra structure on $T_0\mbC^{N+1}$. Let $F^c$ and $F^o$ be the associated solutions to the closed and open WDVV equations. Denote by $\oF=(F^1,\ldots,F^{N+1})$ the vector potential of the associated flat F-manifold. Note that the semisimplicity condition at $0\in\mbC^{N+1}$ implies that $\left.\frac{\d^2 F^o}{(\d t^{N+1})^2}\right|_{t^*=0}\ne 0$.

\medskip

Let $u^1(t^1,\ldots,t^N), \ldots, u^N(t^1,\ldots,t^N)$ be canonical coordinates around $0\in\mbC^N$ for the Dubrovin--Frobenius manifold given by the potential~$F^c$. Making an appropriate shift we can assume that~$u^i(0)=0$.

\medskip

\begin{lemma}\label{lemma:canonical coordinates in the open case}
The functions $u^1,\ldots,u^N$ together with the function $u^{N+1}(t^1,\ldots,t^{N+1}):=\frac{\d F^o}{\d t^{N+1}}$ are canonical coordinates around $0\in\mbC^{N+1}$ for the flat F-manifold given by the pair $(F^c,F^o)$.
\end{lemma}
\begin{proof}
Canonical coordinates of our flat F-manifold can be found as $N+1$ distinct solutions of the system of PDEs
$$
\frac{\d u}{\d t^\alpha}\frac{\d u}{\d t^\beta}=c^\gamma_{\alpha\beta}\frac{\d u}{\d t^\gamma},\quad 1\le\alpha,\beta\le N+1,
$$ 
satisfying $u(0)=0$, where $c^\alpha_{\beta\gamma}:=\frac{\d^2 F^\alpha}{\d t^\beta\d t^\gamma}$. The functions $u^1,\ldots,u^N$ obviously satisfy this system, so it remains to check that $\frac{\d^2 F^o}{\d t^\alpha\d t^{N+1}}\frac{\d F^o}{\d t^\beta\d t^{N+1}}=c^\gamma_{\alpha\beta}\frac{\d^2 F^o}{\d t^\gamma\d t^{N+1}}$, which immediately follows from~\eqref{eq:open WDVV}.
\end{proof}

\medskip

Let us choose the canonical coordinates given by this lemma and follow the construction from Section~\ref{subsubsection:general transitivity}. Clearly, we have $\tPsi^{\le N}_{N+1}=\Psi^{\le N}_{N+1}=\tGamma^{\le N}_{N+1}=\Gamma^{\le N}_{N+1}=0$. Using this, the system of equations~\eqref{eq:R-relation,nondiagonal} immediately implies by induction that $(R_m)^{\le N}_{N+1}=0$. Note that the recursive relations for the entries $(R_{m+1})^i_j$ with $1\le i,j\le N$ involve only the functions $\gamma^k_l$ and~$(R_m)^k_l$ with $1\le k,l\le N$. So we choose the functions $H_1,\ldots,H_N$ such that $\eta=\frac{1}{2}\eta_{\alpha\beta}dt^\alpha dt^\beta=\sum_{i=1}^N H_i^2(du^i)^2$, and then choose the integration constants for $(R_m)^i_i$ with $1\le i\le N$ is such a way that $\pi(R(-z,t^*))^T\pi(R(z,t^*))=\Id$. Note that the system of equations~\eqref{eq:R-relation,diagonal} gives $d(R_m)^{N+1}_{N+1}=0$, and thus we can choose $(R_m)^{N+1}_{N+1}=0$ for $m\ge 1$. It is now obvious that we have $\mcR(z)=\Psi^{-1}(0)R^{-1}(-z,0)\Psi(0)\in G^o_{N+1,+}$. If we denote by $F^c_{\cub}$ the cubic part of $F^c$ and by $F^o_{\quadr}$ the quadratic part of $F^o$, then the element $\mcR(z)\in G^o_{N+1,+}$ transforms the pair of ancestor potentials corresponding to $(F^c_{\cub},F^o_{\quadr})$ to the pair $(\mcF^c_0,\mcF^o_0)$, which proves the proposition.
\end{proof}

\medskip

\subsubsection{Commutation relations for the operators $\widehat{r(z)}^o$ and $\widehat{s(z)}^o$}

For $r(z),\tr(z)\in\mfg^o_{N+1,+}$ let
$$
G_p(r(z),\tr(z)):=\sum_{\substack{i,j\ge 1\\i+j=p}}\left((-1)^{i+1}-(-1)^{j+1}\right)(r_i)^{N+1}_\mu\eta^{\mu\nu}(\tr_j)^{N+1}_\nu,\quad p\ge 2.
$$
Clearly, $G_p(r(z),\tr(z))=0$ if $p$ is even.

\medskip

\begin{lemma}\label{lemma:commutation of open r-operators}
We have the following commutation relations:
\begin{enumerate}[ 1.]
\item $\displaystyle\left[\widehat{r(z)}^o,\widehat{\tr(z)}^o\right]=\widehat{\left[r(z),\tr(z)\right]}^o+\sum_{p\ge 3}G_p(r(z),\tr(z))\left(\eps\frac{\d}{\d q^{N+1}_p}+\frac{\eps^2}{2}\sum_{\substack{i,j\ge 0\\i+j=p-1}}(-1)^{i+1}\frac{\d^2}{\d q^{N+1}_i\d q^{N+1}_j}\right)$, for $r(z),\tr(z)\in\mfg^o_{N+1,+}$;

\smallskip

\item $\displaystyle\left[\widehat{s(z)}^o,\widehat{\ts(z)}^o\right]=\widehat{\left[s(z),\ts(z)\right]}^o$, for $s(z),\ts(z)\in\mfg^o_{N+1,-}$.
\end{enumerate}
\end{lemma}
\begin{proof}
1. One can immediately see that
\begin{align}
&\left[\widehat{r}^o,\widehat{\tr}^o\right]=\widehat{\left[\pi(r),\pi(\tr)\right]}^c\notag\\
&+\sum_{i,j\ge 1}\left(-\sum_{k\ge 0}[r_i,\tr_j]^{N+1}_\mu q^\mu_k\frac{\d}{\d q^{N+1}_{i+j+k}}+\eps(-1)^{i+j}[r_i,\tr_j]^{N+1,\nu}\frac{\d}{\d q^\nu_{i+j}}\right)+\eps\sum_{p\ge 3}G_p(r,\tr)\frac{\d}{\d q^{N+1}_p}\notag\\
&+\eps^2\sum_{i,j\ge 0,\,k\ge 1}\bigg(\underbrace{(-1)^j(\tr_{i+j+1})^{N+1}_\mu\eta^{\mu\nu}(r_k)^\gamma_\nu}_{A^\gamma_{ijk}:=}+\underbrace{(-1)^{j+1}(r_{i+j+1})^{N+1}_\mu\eta^{\mu\nu}(\tr_k)^\gamma_\nu}_{B^\gamma_{ijk}:=}\bigg)\frac{\d^2}{\d q^{N+1}_i\d q^\gamma_{j+k}}\label{line:commutation of r,A and B}\\
&+\eps^2\sum_{i\ge 1,\,j,k\ge 0}\bigg(\underbrace{(-1)^{j+1}(r_{j+k+1})^\alpha_\mu\eta^{\mu\beta}(\tr_i)^{N+1}_\beta}_{C^\alpha_{ijk}:=}+\underbrace{(-1)^j(\tr_{j+k+1})^\alpha_\mu\eta^{\mu\beta}(r_i)^{N+1}_\beta}_{D^\alpha_{ijk}:=}\bigg)\frac{\d^2}{\d q^{N+1}_{i+j}\d q^\alpha_k}\label{line:commutation of r,C and D}\\
&+\eps^2\sum_{i\ge 1,\,j,k\ge 0}\underbrace{(-1)^j\bigg((\tr_{j+k+1})^{N+1}_\mu\eta^{\mu\nu}(r_i)^{N+1}_\nu-(r_{j+k+1})^{N+1}_\mu\eta^{\mu\nu}(\tr_i)^{N+1}_\nu\bigg)}_{E_{ijk}:=}\frac{\d^2}{\d q^{N+1}_{i+j}\d q^{N+1}_k}.\notag
\end{align}
We see that
\begin{align*}
&\sum_{i,j\ge 0,\,k\ge 1}A^\gamma_{ijk}\frac{\d^2}{\d q^{N+1}_i\d q^\gamma_{j+k}}+\sum_{i\ge 1,\,j,k\ge 0}C^\gamma_{ijk}\frac{\d^2}{\d q^{N+1}_{i+j}\d q^\gamma_k}=-\sum_{\substack{a,b\ge 1,\,j,k\ge 0\\a+b=j+k+1}}(-1)^k(\tr_a)^{N+1}_\mu(r_b)^\mu_\nu\eta^{\nu\gamma}\frac{\d^2}{\d q^{N+1}_j\d q^\gamma_k},\\
&\sum_{i,j\ge 0,\,k\ge 1}B^\gamma_{ijk}\frac{\d^2}{\d q^{N+1}_i\d q^\gamma_{j+k}}+\sum_{i\ge 1,\,j,k\ge 0}D^\gamma_{ijk}\frac{\d^2}{\d q^{N+1}_{i+j}\d q^\gamma_k}=\sum_{\substack{a,b\ge 1,\,j,k\ge 0\\a+b=j+k+1}}(-1)^k(r_a)^{N+1}_\mu(\tr_b)^\mu_\nu\eta^{\nu\gamma}\frac{\d^2}{\d q^{N+1}_j\d q^\gamma_k},
\end{align*}
and as a result the sum of the terms in lines~\eqref{line:commutation of r,A and B} and~\eqref{line:commutation of r,C and D} is equal to
$$
\eps^2\sum_{\substack{a,b\ge 1,\,j,k\ge 0\\a+b=j+k+1}}(-1)^k[r_a,\tr_b]^{N+1}_\mu\eta^{\mu\nu}\frac{\d^2}{\d q^{N+1}_j\d q^\nu_k}.
$$
Analogously, we obtain
\begin{align*}
\sum_{i\ge 1,\,j,k\ge 0}E_{ijk}\frac{\d^2}{\d q^{N+1}_{i+j}\d q^{N+1}_k}=&\sum_{\substack{a,b\ge 1,\,i,j\ge 0\\a+b=i+j+1}}(-1)^{a+i}(r_a)^{N+1}_\mu\eta^{\mu\nu}(\tr_b)^{N+1}_\nu\frac{\d^2}{\d q^{N+1}_i\d q^{N+1}_j}=\\
=&\frac{1}{2}\sum_{i,j\ge 0}(-1)^{i+1}G_{i+j+1}(r,\tr)\frac{\d^2}{\d q^{N+1}_i\d q^{N+1}_j},
\end{align*}
which completes the proof of Part 1.

\medskip

2. The proof is easy and we omit it.
\end{proof}

\medskip

\subsubsection{Some explicit formulas in canonical coordinates}

Here, in addition to Lemma~\ref{lemma:canonical coordinates in the open case}, we would like to present explicit formulas for the functions~$(R_m)^{N+1}_i$ and $H_{N+1}$ appearing at the end of the proof of Proposition~\ref{proposition:open transitivity}.

\medskip

\begin{lemma}
Using the notations from the last paragraph of the proof of Proposition~\ref{proposition:open transitivity}, we have the following statements.
\begin{enumerate}[ 1.]
\item Up to the multiplication by a nonzero constant, we have $H_{N+1}=\left(\frac{\d^2 F^o}{(\d t^{N+1})^2}\right)^{-1}$.

\smallskip

\item $(R_m)^{N+1}_j=(-1)^{m-1}\frac{\d^{m-1}\gamma^{N+1}_j}{(\d u^{N+1})^{m-1}}$ for $m\ge 1$ and $1\le j\le N$.
\end{enumerate}
\end{lemma}
\begin{proof}
The function $H_{N+1}$ is determined by the equation $H_{N+1}^{-1}\cdot d H_{N+1}=-(d\tPsi\cdot\tPsi^{-1})^{N+1}_{N+1}$. Since the matrix $\tPsi$ has the form
\[
\widetilde \Psi = 
\left(
\begin{array}{c|c}
* & 
         {\begin{array}{c}
          0 \\ \vdots \\ 0
         \end{array}}
\\
\hline
\frac{\d^2 F^o}{\d t^1\d t^{N+1}} \dots \frac{\d^2 F^o}{\d t^N\d t^{N+1}} & \frac{\d^2 F^o}{(\d t^{N+1})^2} 
\end{array}
\right),
\]
Part 1 of the lemma is proved.

\medskip

For Part 2, first note that equation~\eqref{eq:R-relation,nondiagonal} immediately gives that $(R_1)^i_j=\gamma^i_j$ for $i\ne j$. Since we chose $(R_m)^{N+1}_{N+1}=0$ for $m\ge 1$, equation~\eqref{eq:R-relation,nondiagonal} gives $(R_{m+1})^{N+1}_j=-\frac{\d(R_m)^{N+1}_j}{\d u^{N+1}}$ for $m\ge 1$, which completes the proof.
\end{proof}

\medskip

\section{Open total descendant potentials}\label{section:open total descendant potentials}

In this section, after recalling the Givental construction of closed total descendant potentials, we present our construction of open total descendant potentials, discuss group actions on the space of these potentials, and prove Theorem~\ref{theorem:main}.

\medskip

We fix $N\ge 1$.

\medskip 

\subsection{The Givental construction of closed total descendant potentials}

Let $\oA\in\mbC^N$ be a nonzero vector and $\eta=(\eta_{\alpha\beta})$ be an $N\times N$ constant symmetric nondegenerate matrix. By Givental, the space of \emph{closed total ancestor potentials} of rank $N$,
\begin{gather*}
\mcF^{c,\anc}(t^1_*,\ldots,t^N_*,\eps)=\sum_{g\ge 0}\eps^{2g-2}\mcF^{c,\anc}_g(t^1_*,\ldots,t^N_*),\quad \mcF^{c,\anc}_g\in\mbC[[t^*_*]],
\end{gather*}
is defined by
\begin{gather*}
\exp(\mcF^{c,\anc}):=\exp\left(\widehat{r(z)}^c\right)\exp\left(\widehat{\psi}\right)\left(\prod_{i=1}^N\tau^{\KW}(a_i t^i_*,a_i\eps)\right),
\end{gather*}
where the parameters $r(z)\in\mfg^c_{N,+}$, $a_1,\ldots,a_N\in\mbC^*$, and $\psi\in\End(\mbC^N)$ satisfy the conditions 
\begin{gather}\label{eq:relations for closed parameters}
\exp(\psi)^i_\alpha A^\alpha= a_i^{-1},\qquad \exp(\psi)_\alpha^i\eta^{\alpha\beta}\exp(\psi)^j_\beta=\delta^{ij}.
\end{gather}

\medskip

Then the space of \emph{closed total descendant potentials} $\mcF^{c,\desc}(t^*_*,\eps)=\sum_{g\ge 0}\eps^{2g-2}\mcF^{c,\desc}_g(t^*_*)$ of rank~$N$ is given by
\begin{gather*}
\exp(\mcF^{c,\desc}):=\exp\left(\widehat{s(z)}^c\right)\exp(\mcF^{c,\anc}),
\end{gather*}
where $s(z)\in\mfg^c_{N,-}$ and $\mcF^{c,\anc}$ is an arbitrary closed total ancestor potential. Note that in general the resulting potentials $\mcF^{c,\desc}_g$ are formal power series in shifted variables: $\mcF^{c,\desc}_g \in\mcR^{-s_1\oA}_N\mbC[[t^*_{\ge 1}]]$. 

\medskip

Let us make some remarks about this construction. First of all, note that since the maps $r(z)\mapsto\widehat{r(z)}^c$ and $s(z)\mapsto\widehat{s(z)}^c$ are Lie algebra homomorphisms from the Lie algebras $\mfg^c_{N,+}$ and~$\mfg^c_{N,-}$, respectively, to the Lie algebra of differential operators, the Givental construction gives a $G^c_{N,+}$-action on the space of closed total ancestor potentials and a $G^c_{N,-}$-action on the space of closed total descendant potentials.

\medskip

Then note that the function $\sum_{i=1}^N a_i^{-2}\mcF^\KW_0(a_i t^i_*)$ is the closed ancestor potential corresponding to the solution $\sum_{i=1}^Na_i\frac{(t^i)^3}{6}$ of the closed WDVV equations, with the metric $(\delta_{ij})$ and the unit $\sum_{i=1}^Na_i^{-1}\frac{\d}{\d t^i}$. The relations~\eqref{eq:relations for closed parameters} imply that the function $\exp\left(\widehat{\psi}\right)\left(\sum_{i=1}^N a_i^{-2}\mcF^\KW_0(a_i t^i_*)\right)$ is a closed ancestor potential with the metric $\eta$ and the unit $A^\alpha\frac{\d}{\d t^\alpha}$. Formula~\eqref{eq:closed r-action,difop} (respectively, \eqref{eq:closed s-action,difop}) shows that the $G^c_{N,+}$-action (respectively, $G^c_{N,-}$-action) on the space of closed total ancestor (respectively, descendant) potentials agrees with the $G^c_{N,+}$-action (respectively, $G^c_{N,-}$-action) on the space of closed ancestor (respectively, descendant) potentials $\mcF^c_0$ discussed in Section~\ref{subsubsection:closed actions in genus 0}.

\medskip

Remembering also Section~\ref{subsubsection:semisimplicity}, we see that the functions $\mcF^{c,\desc}_0$ are closed descendant potentials defining a semisimple algebra structure on $T_{\ot_\orig}\mbC^N$. Conversely, for any closed descendant potential $\mcF^c_0$ defining a semisimple algebra structure on $T_{\ot_\orig}\mbC^N$ there exists a closed total descendant potential $\mcF^{c,\desc}$ such that $\mcF^{c,\desc}_0=\mcF^c_0$.

\medskip

The following equations are the closed analogs of equations~\eqref{eq:open TRR1},~\eqref{eq:open string equation}, and~\eqref{eq:open dilaton equation}, respectively (the closed analog of equation~\eqref{eq:open TRR0} is equation~\eqref{eq:closed TRR-0}):
\begin{align}
&\frac{\d\mcF_1^{c,\desc}}{\d t^\alpha_{a+1}}=\frac{\d^2\mcF^{c,\desc}_0}{\d t^\alpha_a\d t^\mu_0}\eta^{\mu\nu}\frac{\d\mcF^{c,\desc}_1}{\d t^\nu_0}+\frac{1}{24}\eta^{\mu\nu}\frac{\d^3\mcF_0^{c,\desc}}{\d t^\alpha_a\d t_0^\mu\d t_0^\nu}, && 1\le\alpha\le N,\quad a\ge 0,\notag\\
&A^\alpha\frac{\d\mcF^{c,\desc}}{\d t^\alpha_0}=\sum_{d\ge 0}t^\alpha_{d+1}\frac{\d\mcF^{c,\desc}}{\d t^\alpha_d}+\frac{\eps^{-2}}{2}\eta_{\alpha\beta}t^\alpha_0 t^\beta_0+\frac{1}{2}\mathrm{tr}(r_1),&&\label{eq:closed string equation}\\
&A^\alpha\frac{\d\mcF^{c,\desc}}{\d t^\alpha_1}=\sum_{d\ge 0}t^\alpha_d\frac{\d\mcF^{c,\desc}}{\d t^\alpha_d}+\eps\frac{\d\mcF^{c,\desc}}{\d\eps}+\frac{N}{24}.&&\label{eq:closed dilaton equation}
\end{align}

\medskip

\subsection{A construction of open total descendant potentials}

Recall~\cite{PST14,Bur15} that 
\begin{gather}\label{eq:dimension condition for PST}
\begin{minipage}{12cm}
the coefficient of $\prod_{i=1}^l t_{d_i}\prod_{j=1}^k s_{m_j}$ in the formal power series $\mcF^\PST_g$ is zero unless $2\sum_{i=1}^l d_i+2\sum_{j=1}^k m_j=3g-3+k+2l$.
\end{minipage}
\end{gather}
This implies that 
$$
\mcF^\PST_g\in\mbC[s_0][[t_*,s_{\ge 1}]].
$$
Therefore, for any $\theta\in\mbC$ the transformation
\begin{align}
\mcF^\PST\mapsto \mcF^\PST_{(\theta)}:=&\left.\mcF^\PST\right|_{s_i\mapsto s_i+\delta_{i,0}\theta+\sum_{j\ge 1}(-1)^j\left(\frac{\theta^{2j}}{2^j j!}s_{i+j}+\frac{\theta^{2j-1}}{(2j-1)2^{j-1}(j-1)!}t_{i+j}\right)}+\label{eq:shifting by theta}\\
&+\eps^{-1}\sum_{j\ge 1}(-1)^j\left(\frac{\theta^{2j}}{2^j j!}s_{j-1}+\frac{\theta^{2j-1}}{(2j-1)2^{j-1}(j-1)!}t_{j-1}\right)-\eps^{-1}\frac{\theta^3}{6}\notag
\end{align}
is well defined, and $\mcF^\PST_{(\theta)}$ has the form $\mcF^\PST_{(\theta)}=\sum_{g\ge 0}\eps^{g-1}\mcF^\PST_{(\theta),g}$, where $\mcF^\PST_{(\theta),g}\in\mbC[[t_*,s_*]]$. We define $\tau^\PST_{(\theta)}:=\exp(\mcF^\PST_{(\theta)})$.

\medskip 

We fix an $N\times N$ constant symmetric nondegenerate matrix $\eta=(\eta_{\alpha\beta})$ and a vector $\oA=(A^1,\ldots,A^{N+1})\in\mbC^{N+1}$ such that $\pi(\oA)\ne 0\in\mbC^N$. Consider the following data: $r(z)\in\mfg^o_{N+1,+}$, $a_1,\ldots,a_N,\theta\in\mbC^*$, $\ob=(b_1,\ldots,b_N)\in\mbC^N$, and $\psi\in\End(\mbC^N)$ satisfying
$$
\exp(\psi)_\alpha^i A^\alpha=a_i^{-1},\qquad \exp(\psi)_\alpha^i\eta^{\alpha\beta}\exp(\psi)_\beta^j=\delta^{ij},\qquad b_\alpha A^\alpha+A^{N+1}=0.
$$
We then have the closed total ancestor potential $\mcF^{c,\anc}$ corresponding to the parameters $\pi(r(z))\in\mfg^c_{N,+}$, $a_1,\ldots,a_N\in\mbC^*$, and $\psi\in\End(\mbC^N)$. We define an \emph{open total ancestor potential}  
\begin{gather*}
\mcF^{o,\anc}(t^1_*,\ldots,t^{N+1}_*,\eps)=\sum_{g\ge 0}\eps^{g-1}\mcF^{o,\anc}_g(t^1_*,\ldots,t^{N+1}_*),\quad \mcF^o_g\in\mbC[[t^*_*]],
\end{gather*}
by
\begin{gather}\label{eq:open ancestor potential}
\exp(\mcF^{o,\anc}+\mcF^{c,\anc}):=\exp\left(\widehat{r(z)}^o\right)\exp\left(\widehat{\ob}\right)\exp\left(\widehat{\psi}\right)\left(\tau^{\PST}_{(\theta)}(a_1 t^1_*,t^{N+1}_*,\eps)\prod_{i=1}^N\tau^{\KW}(a_i t^i_*,a_i\eps)\right).
\end{gather}
Note that there is no rescaling of $\eps$ in $\tau^{\PST}_{(\theta)}$ in this formula. If in addition we have an element $s(z)\in\mfg^o_{N+1,-}$, then there is the corresponding closed total descendant potential $\mcF^{c,\desc}$ given by
$$
\exp(\mcF^{c,\desc})=\exp\left(\widehat{\pi(s(z))}^c\right)\exp(\mcF^{c,\anc}),
$$
and we define an \emph{open total descendant potential} $\mcF^{o,\desc}=\sum_{g\ge 0}\eps^{g-1}\mcF^{o,\desc}_g$ by
$$
\exp(\mcF^{o,\desc}+\mcF^{c,\desc})=\exp\left(\widehat{s(z)}^o\right)\exp(\mcF^{o,\anc}+\mcF^{c,\anc}).
$$
Note that $\mcF^{o,\desc}_g\in\mcR_{N+1}^{-s_1\oA}[[t^*_{\ge 1}]]$. 

\medskip

\subsection{First properties}\label{subsection:first properties}

We see that a pair of total descendant potentials $(\mcF^{c,\desc},\mcF^{o,\desc})$, with metric $\eta$ and unit $\oA\in\mbC^{N+1}$ such that $\pi(\oA)\ne 0$, is associated to the data
$$
a_1,\ldots,a_N,\theta\in\mbC^*,\quad\psi\in\End(\mbC^N),\quad\ob\in\mbC^N,\quad r(z)\in\mfg^o_{N+1,+},\quad s(z)\in\mfg^o_{N+1,-},
$$
satisfying the conditions
\begin{gather}\label{eq:relations for parameters}
\exp(\psi)_\alpha^i A^\alpha=a_i^{-1},\qquad \exp(\psi)_\alpha^i\eta^{\alpha\beta}\exp(\psi)_\beta^j=\delta^{ij},\qquad b_\alpha A^\alpha+A^{N+1}=0.
\end{gather}
Denote by $\Desc_N$ the set of all pairs of total descendant potentials $(\mcF^{c,\desc},\mcF^{o,\desc})$. Choosing $s(z)=0$, we obtain the set of total ancestor potentials, which we denote by $\Anc_N\subset\Desc_N$. Further choosing $r(z)=0$, we obtain a smaller subset of $\Anc_N$, which we denote by $\Anc_N^1$. Choosing also $\ob=0$, we get a subset of $\Anc^1_N$ denoted by $\Anc^0_N$.

\medskip

Let us first discuss the initial pair of total ancestor potentials 
$$
\left(\sum_{i=1}^N\mcF^\KW(a_i t^i_*,a_i\eps),\mcF^\PST_{(\theta)}(a_1t^1_*,t^{N+1}_*,\eps)\right)\in\Anc_N^0.
$$
Note that the functions
$$
\mcF^c_0=\sum_{i=1}^N a_i^{-2}\mcF^\KW_0(a_i t^i_*),\qquad \mcF^o_0=\mcF^\PST_0(a_1 t^1_*,t^{N+1}_*)
$$
are the closed and open ancestor potentials corresponding to the solutions 
\begin{gather*}\label{eq:initial closed and open solutions}
F^c=\sum_{i=1}^Na_i\frac{(t^i)^3}{6},\qquad F^o=a_1 t^1 t^{N+1}+\frac{(t^{N+1})^3}{6}
\end{gather*}
to the closed and open WDVV equations, with the metric $(\delta_{ij})$ and the unit $\sum_{i=1}^N a_i^{-1}\frac{\d}{\d t^i}$. The corresponding algebra structure on $T_0\mbC^{N+1}$ is not semisimple. However, note that for any $\theta\ne 0$ the algebra structure on the tangent space at $(0,\ldots,0,\theta)\in\mbC^{N+1}$ is semisimple. Consider the matrices $\Omega_k$ defined in Section~\ref{subsubsection:shifts and lower action}. It is easy to compute that
$$
\Omega_k(0,\ldots,0,\theta)=\begin{pmatrix}
0      &   & \ldots &   &  0    \\
\vdots &   &        &   &\vdots \\
0      &   & \ldots &   & 0     \\
\frac{a_1\theta^{2k+1}}{(2k+1)2^k k!} & 0 & \ldots & 0 & \frac{\theta^{2k+2}}{2^{k+1}(k+1)!}
\end{pmatrix}.
$$
Using Lemma~\ref{lemma:exponent of open s}, we see that
\begin{gather}\label{eq:theta-shift and s}
\tau^{\PST}_{(\theta)}(a_1 t^1_*,t^{N+1}_*,\eps)\prod_{i=1}^N\tau^{\KW}(a_i t^i_*,a_i\eps)=\exp\left(\widehat{s(z)}^o\right)\left(\tau^{\PST}(a_1 t^1_*,t^{N+1}_*,\eps)\prod_{i=1}^N\tau^\KW(a_i t^i_*,a_i\eps)\right),
\end{gather}
where 
\begin{gather}\label{eq:s and Omega for theta}
s(z):=\log S(z),\qquad S(z):=\left(\Id+\sum_{j\ge 1}(-1)^j\Omega_{j-1}z^{-j}\right)^{-1}.
\end{gather}
Formulas~\eqref{eq:open infinitesimal s-action} and~\eqref{eq:ancestor vector potentials at shifted points} then imply that $\mcF^o_{(\theta),0}$ is the open ancestor potential corresponding to the solutions 
$$
F^c=\sum_{i=1}^Na_i\frac{(t^i)^3}{6},\qquad F^o_{(\theta)}:=F^o|_{t^{N+1}\mapsto t^{N+1}+\theta}=a_1 t^1 (t^{N+1}+\theta)+\frac{(t^{N+1}+\theta)^3}{6}
$$
to the closed and open WDVV equations. The corresponding algebra structure on $T_0\mbC^{N+1}$ is semisimple.

\medskip

Let us now discuss the spaces $\Anc^0_N$ and $\Anc^1_N$. For $H=\widehat{\psi}$, $\psi\in\End(\mbC^N)$, the transformation
\begin{gather*}
\left(\exp\left(\mcF^{c,\anc}\right),\exp\left(\mcF^{c,\anc}+\mcF^{o,\anc}\right)\right)\mapsto\left(\exp\left(-H\right)\exp\left(\mcF^{c,\anc}\right),\exp\left(-H\right)\exp\left(\mcF^c+\mcF^{o,\anc}\right)\right)
\end{gather*}
defines a $\GL(\mbC^N)$-action on $\Anc^0_N$, while for $H=\widehat{\ov}$, $\ov\in\mbC^N$, this formula defines a $\mbC^N$-action on $\Anc^1_N$ (the last claim is justified by noting that $\left[\widehat{\overline{v}},\widehat{\overline{w}}\right]=0$ for any $\overline{v},\overline{w}\in\mbC^N$). At the level of functions $\mcF^{c,\anc}_0$ and~$\mcF^{o,\anc}_0$, the formulas for these actions coincide with the formulas~\eqref{eq:hatm-action} and~\eqref{eq:hatov-action} for the actions of these groups on the space of pairs $(\mcF^c_0,\mcF^o_0)$ of closed and open ancestor potentials. We conclude that for any pair $(\mcF^{c,\anc},\mcF^{o,\anc})\in\Anc^1_N$ the function $\mcF^{o,\anc}_0$ is an open ancestor potential, and by~\eqref{eq:relations for parameters} the corresponding metric is $\eta$ and the unit is $A^\alpha\frac{\d}{\d t^\alpha}$. 

\medskip

Let us now discuss the space of all pairs of total ancestor potentials $\Anc_N$.

\medskip

\begin{proposition}\label{proposition:sk-property}
For any pair $(\mcF^{c,\anc},\mcF^{o,\anc})\in\Anc_N$ we have 
\begin{gather}\label{eq:sk-property}
\frac{\d}{\d t^{N+1}_k}\exp(\mcF^{o,\anc})=\frac{\eps^k}{(k+1)!}\frac{\d^{k+1}}{(\d t^{N+1}_0)^{k+1}}\exp(\mcF^{o,\anc}),\quad k\ge 0.
\end{gather}
\end{proposition}
\begin{proof}
Denote
$$
O_k:=\frac{\d}{\d t^{N+1}_k}-\frac{\eps^k}{(k+1)!}\frac{\d^{k+1}}{(\d t^{N+1}_0)^{k+1}},\quad k\ge 0.
$$
From formula~\eqref{eq:open ancestor potential} and the fact that the operator $O_k$ obviously commutes with the operators $\widehat{r(z)}^o$, $\widehat{\psi}$, and $\widehat{\ob}$, it follows that it is sufficent to prove equation~\eqref{eq:sk-property} only for $\mcF^{o,\anc}=\mcF^\PST_{(\theta)}(a_1t^1_*,t^{N+1}_*,\eps)$.

\medskip 

\begin{lemma}
Consider an element $s(z)\in\mfg^o_{N+1,-}$ such that $(s_{\ge 2})^{N+1}_{N+1}=0$. Then we have
$$
O_k\circ\exp\left(\widehat{s(z)}^o\right)=\exp\left(\widehat{s(z)}^o\right)\circ\sum_{j=0}^k\frac{\left(-(s_1)^{N+1}_{N+1}\right)^j}{j!}O_{k-j},\quad k\ge 0.
$$
\end{lemma}
\begin{proof}
For any operators $X$, $Y$ and a formal variable $t$ we have the identity
\begin{gather}\label{eq:exp(X) and Y}
\exp(tX)\circ Y\circ\exp(-tX)=Y+\sum_{n\ge 1}\frac{t^n}{n!}\ad_X^n Y.
\end{gather}
Substituting~$X=\widehat{s(z)}^o$, $Y=O_k$, and $t=-1$ in~\eqref{eq:exp(X) and Y} and noting that $\left[\widehat{s(z)}^o,O_k\right]=(s_1)^{N+1}_{N+1}O_{k-1}$ for $k\ge 0$ (we adopt the convention $O_{-1}:=0$), we obtain the desired result.
\end{proof}

\medskip

Equation~\eqref{eq:sk-property} for $\mcF^{o,\anc}=\mcF^\PST_{(\theta)}(a_1t^1_*,t^{N+1}_*,\eps)$ follows from this lemma, the property (see Remark~\ref{remark:descendant of s})
$$
O_j\exp\left(\mcF^\PST(a_1t^1_*,t^{N+1}_*,\eps)\right)=0,
$$
and equation~\eqref{eq:theta-shift and s}, where one should note that the property $(s_{\ge 2})^{N+1}_{N+1}=0$ is satisfied for $s(z)$ given by~\eqref{eq:s and Omega for theta}, because $(S_k)^{N+1}_{N+1}=\frac{1}{k!}\left((S_1)^{N+1}_{N+1}\right)^k$.
\end{proof}

\medskip

\begin{lemma}
For any $r(z),\tr(z)\in\mfg^o_{N+1,+}$ and $(\mcF^{c,\anc},\mcF^{o,\anc})\in\Anc_N$ we have
$$
\left(\left[\widehat{r(z)}^o,\widehat{\tr(z)}^o\right]-\widehat{\left[r(z),\tr(z)\right]}^o\right)\exp(\mcF^{c,\anc}+\mcF^{o,\anc})=0.
$$
\end{lemma}
\begin{proof}
By Lemma~\ref{lemma:commutation of open r-operators}, it is sufficient to check that the expression
$$
\eps\frac{\d\exp(\mcF^{o,\anc})}{\d t^{N+1}_p}+\frac{\eps^2}{2}\sum_{\substack{i,j\ge 0\\i+j=p-1}}(-1)^{i+1}\frac{\d^2\exp(\mcF^{o,\anc})}{\d t^{N+1}_i\d t^{N+1}_j}
$$
vanishes for any odd $p$. Indeed, by Proposition~\ref{proposition:sk-property} this expression is equal to
\begin{align*}
\Bigg(1+\frac{1}{2}\sum_{\substack{k,l\ge 1\\k+l=p+1}}(-1)^k{p+1\choose k}\Bigg)\frac{\eps^{p+1}}{(p+1)!}\frac{\d^{p+1}\exp(\mcF^{o,\anc})}{(\d t^{N+1}_0)^{p+1}}=0,
\end{align*}
as required.
\end{proof}

\medskip

The lemma immediately implies that the formula
\begin{align*}
&(\exp(\mcF^{c,\anc}),\exp(\mcF^{c,\anc}+\mcF^{o,\anc}))\mapsto\\
&\hspace{1.8cm}\mapsto\left(\exp\left(\widehat{\pi(r(z))}^c\right)\exp(\mcF^{c,\anc}),\exp\left(\widehat{r(z)}^o\right)\exp(\mcF^{c,\anc}+\mcF^{o,\anc})\right),\quad r(z)\in\mfg^o_{N+1,+},
\end{align*}
defines a $G^o_{N+1,+}$-action on the space $\Anc_N$. At the level of functions $\mcF^{c,\anc}_0$ and~$\mcF^{o,\anc}_0$, the formula for this action coincides with formula~\eqref{eq:open infinitesimal r-action} for the $G^o_{N+1,+}$-action on the space of pairs $(\mcF^c_0,\mcF^o_0)$ of closed and open ancestor potentials. We thus see that $\mcF^{o,\anc}_0$ is an open ancestor potential.

\medskip

Regarding the whole space $\Desc_N$, Lemma~\ref{lemma:commutation of open r-operators} immediately implies that the formula
\begin{align*}
&\left(\exp(\mcF^{c,\desc}),\exp(\mcF^{c,\desc}+\mcF^{o,\desc})\right)\mapsto\\
&\hspace{1.8cm}\mapsto\left(\exp\left(\widehat{\pi(s(z))}^c\right)\exp(\mcF^{c,\desc}),\exp\left(\widehat{s(z)}^o\right)\exp(\mcF^{c,\desc}+\mcF^{o,\desc})\right),\quad s(z)\in\mfg^o_{N+1,-},
\end{align*}
defines a $G^o_{N+1,-}$-action on the space $\Desc_N$. Under this action, the point $\ot_\orig$ changes as follows:
$$
\ot_{\orig}\mapsto\ot_\orig-s_1\oA.
$$ 
At the level of functions $\mcF^{c,\desc}_0$ and $\mcF^{o,\desc}_0$ the formula for the action coincides with~\eqref{eq:open infinitesimal s-action}. Therefore, $\mcF^{o,\desc}_0$ is an open descendant potential.

\medskip

\subsection{Proof of Theorem~\ref{theorem:main}}

In Section~\ref{subsection:first properties} we proved that for any $(\mcF^{c,\desc},\mcF^{o,\desc})\in\Desc_N$ the pair $(\mcF^{c,\desc}_0,\mcF^{o,\desc}_0)$ is a pair of closed an open descendant potentials. This implies Part 1 of the theorem.

\medskip

Part 2 is a corollary of the following result.

\medskip

\begin{proposition}
For any pair of closed an open descendant potentials $(\mcF^c_0,\mcF^o_0)$ defining a semisimple algebra structure on $T_{\ot_\orig}\mbC^{N+1}$ there exists a pair $(\mcF^{c,\desc},\mcF^{o,\desc})\in\Desc_N$ such that $\mcF^{c,\desc}_0=\mcF^c_0$ and $\mcF^{o,\desc}_0=\mcF^o_0$.
\end{proposition}
\begin{proof}
It is easy to see that there exists an $(N+1)\times(N+1)$ matrix $s_1$ such that $s_1\oA=\ot_\orig$ and $s_1z^{-1}\in\mfg^o_{N+1,-}$. Acting on the pair $(\mcF^c_0,\mcF^o_0)$ by the element $\exp(s_1z^{-1})\in G^o_{N+1,-}$ we shift the point~$\ot_\orig$ to $0\in\mbC^{N+1}$. Therefore, without loss of generality we can assume that $\ot_\orig=0$.

\medskip

Define $(N+1)\times(N+1)$ matrices $\Omega_j=(\Omega_{j;\beta}^\alpha(t^*))$, $j\ge 0$, by $\Omega^\alpha_{j;\beta}:=\left.\frac{\d\mcF^{\alpha,j}_0}{\d t^\beta}\right|_{t^*_{\ge 1}=0}$, where $\omcF^a=(\mcF^{1,a},\ldots,\mcF^{N+1,a})$ are the descendant vector potentials corresponding to $(\mcF^c_0,\mcF^o_0)$. By~\cite[Proposition 2.10]{ABLR20}, the element $\left(\Id+\sum_{j\ge 1}(-1)^jz^{-j}\Omega_{j-1}(0)\right)^{-1}\in G^o_{N+1,-}$ transforms the pair $(\mcF^c_0,\mcF^o_0)$ to a pair of ancestor potentials. Therefore, without loss of generality we can further assume that our pair of potentials $(\mcF^c_0,\mcF^o_0)$ is ancestor. 

\medskip

Denote by $F^c$ and $F^o$ the corresponding solutions to the closed and open WDVV equations. Let us consider the canonical coordinates on the flat F-manifold given by $(F^c,F^o)$ described by Lemma~\ref{lemma:canonical coordinates in the open case}. The metric of the Dubrovin--Frobenius manifold becomes diagonal in the canonical coordinates: $\frac{1}{2}\eta_{\alpha\beta}dt^\alpha dt^\beta=\sum_{i=1}^N g_i(du^i)^2$, where we consider $g_i$ as a function of $t^*$, $g_i=g_i(t^*)$. It is easy to check that the pair of total ancestor potentials from the space $\Anc_N^1$ corresponding to the parameters $a_1,\ldots,a_N,\theta\in\mbC^*$, $b_1,\ldots,b_N\in\mbC^N$, and $\psi\in\End(\mbC^N)$ defined by
$$
a_i:=\sqrt{g_i(0)}^{-1},\quad\theta:=\frac{\d u^{N+1}}{\d t^{N+1}}(0),\quad \exp(\psi)^i_\alpha:=a_i^{-1}\frac{\d u^i}{\d t^\alpha}(0),\quad b_\alpha:=\theta^{-1}\left(\frac{\d u^{N+1}}{\d t^\alpha}(0)-\frac{\d u^1}{\d t^\alpha}(0)\right),
$$
gives a flat F-manifold having the same algebra structure at the origin as the flat F-manifold given by $(F^c,F^o)$. By Proposition~\ref{proposition:open transitivity}, applying an appropriate element of the group $G^o_{N+1,+}$ we obtain a pair $(\mcF^{c,\anc},\mcF^{o,\anc})\in\Anc_N$ such that $\mcF^{c,\anc}_0=\mcF^c_0$ and $\mcF^{o,\anc}_0=\mcF^o_0$.
\end{proof}

\medskip


\section{Proof of Theorem~\ref{theorem:main2}}\label{section:proof of theorem-2}

\subsection{The open TRR-$0$ relations}

In Section~\ref{subsection:first properties} we proved that $\mcF^{o,\desc}_0$ is an open descendant potential and in Section~\ref{subsubsection:definition of open descendant} we showed that any open descendant potential satisfies the open TRR-$0$ equations. So we are done with this part.

\medskip

\subsection{The open TRR-$1$ relations}

First of all, the system of equations~\eqref{eq:open TRR1} is satisfied for $N=1$, $\mcF^c_g(t^1_*)=\mcF^\KW_g(t^1_*)$, and $\mcF^o_g=\mcF^\PST_g(t^1_*,t^2_*)$ \cite[Section~6.2.3]{BCT18}. Then it is easy to see that equation~\eqref{eq:open TRR1} holds for an arbitrary $N\ge 1$ and $\mcF^c_g=\sum_{i=1}^N a_i^{2g-2}\mcF^\KW_g(a_i t^i_*)$, $\mcF^o_g=\mcF^\PST_g(a_1 t^1_*,t^{N+1}_*)$.

\medskip

Let us now prove that equation~\eqref{eq:open TRR1} is preserved by the $G_{N+1,-}^o$-action on $\Desc_N$. It is sufficient to check that the group $G^o_{N+1,-}$ preserves equation~\eqref{eq:open TRR1} infinitesimally,~i.e.
\begin{gather}\label{eq:open TRR1,inf}
\frac{\d(\delta\mcF_1^o)}{\d t^\alpha_{a+1}}=\frac{\d^2(\delta\mcF^c_0)}{\d t^\alpha_a\d t^\mu_0}\eta^{\mu\nu}\frac{\d\mcF^o_1}{\d t^\nu_0}+\frac{\d^2\mcF^c_0}{\d t^\alpha_a\d t^\mu_0}\eta^{\mu\nu}\frac{\d(\delta\mcF^o_1)}{\d t^\nu_0}+\frac{\d(\delta\mcF^o_0)}{\d t^\alpha_a}\frac{\d\mcF^o_1}{\d t_0^{N+1}}+\frac{\d\mcF^o_0}{\d t^\alpha_a}\frac{\d(\delta\mcF^o_1)}{\d t_0^{N+1}}+\frac{1}{2}\frac{\d^2(\delta\mcF_0^o)}{\d t^\alpha_a\d t_0^{N+1}},
\end{gather}
where $\delta\bullet:=s(z).\bullet$ and $s(z)\in\mfg^o_{N+1,-}$. We have 
\begin{gather*}
\delta\mcF^c_0=-\widehat{s(z)}\mcF^c_0+\frac{1}{2}\sum_{i,j\ge 0}(-1)^{j+1}(s_{i+j+1})_\alpha^\mu\eta_{\mu\beta}q^\alpha_i q^\beta_j,\qquad \delta\mcF^o_g=-\widehat{s(z)}\mcF^o_g-\delta_{g,0}\sum_{i\ge 1}(s_i)^{N+1}_\alpha q^\alpha_{i-1},
\end{gather*}
where we recall our notation $\widehat{s(z)}=\sum_{i\ge 1,\,j\ge 0}(s_i)^\gamma_\beta q^{\beta}_{i+j}\frac{\d}{\d q^\gamma_j}$. Using the equation obtained from~\eqref{eq:open TRR1} by applying the operator $-\widehat{s(z)}$ to both sides, we see that equation~\eqref{eq:open TRR1,inf} is equivalent to
\begin{align*}
-\sum_{\substack{i\ge 1,\,j\ge 0\\i+j=a+1}}(s_i)^\gamma_\alpha\frac{\d\mcF^o_1}{\d t^\gamma_j}=&-\sum_{\substack{i\ge 1,\,j\ge 0\\i+j=a}}(s_i)_\alpha^\gamma\left(\frac{\d^2\mcF^c_0}{\d t^\gamma_j\d t^\mu_0}\eta^{\mu\nu}\frac{\d\mcF^o_1}{\d t^\nu_0}+\frac{\d\mcF^o_0}{\d t^\gamma_j}\frac{\d\mcF^o_1}{\d t^{N+1}_0}+\frac{1}{2}\frac{\d^2\mcF^o_0}{\d t^\gamma_j\d t^{N+1}_0}\right)-(s_{a+1})^\nu_\alpha\frac{\d\mcF^o_1}{\d t^\nu_0}.
\end{align*}
Applying equation~\eqref{eq:open TRR1} to the terms with $j\ge 1$ on the left-hand side, we immediately see that this equation is true. In particular, we obtain that equation~\eqref{eq:open TRR1} is true for the pair $\left(\sum_{i=1}^N\mcF^{\KW}(a_i t^i_*,a_i\eps),\mcF^\PST_{(\theta)}(a_1 t^1_*,t^{N+1}_*)\right)$

\medskip

The fact that the $\GL(\mbC^N)$-action on $\Anc^0_N$ preserves equation~\eqref{eq:open TRR1} is obvious. Therefore, equation~\eqref{eq:open TRR1} is true for any pair $(\mcF^{c,\anc},\mcF^{o,\anc})\in\Anc_N^0$.

\medskip

Let us prove that equation~\eqref{eq:open TRR1} is preserved by the $\mbC^N$-action on $\Anc^1_N$. We again proceed infinitesimally,~i.e we want to prove that equation~\eqref{eq:open TRR1,inf} with $\delta\bullet:=\ob.\bullet$ and $\ob\in\mbC^N$ is satisfied assuming that equation~\eqref{eq:open TRR1} is true. For this we compute
$$
\delta\mcF^c_0=0,\qquad \delta\mcF^o_0=-\sum_{d\ge 0}b_\alpha t^\alpha_d\frac{\d\mcF^o_0}{\d t^{N+1}_d}+b_\alpha\eta^{\alpha\beta}\frac{\d\mcF^c_0}{\d t^\beta_0},\qquad \delta\mcF^o_1=-\sum_{d\ge 0}b_\alpha t^\alpha_d\frac{\d\mcF^o_1}{\d t^{N+1}_d}+b_\alpha\eta^{\alpha\beta}\frac{\d\mcF^o_0}{\d t^\beta_0}.
$$
Using the equation obtained from~\eqref{eq:open TRR1} by applying the operator $-\sum_{d\ge 0}b_\alpha t^\alpha_d\frac{\d}{\d t^{N+1}_d}$ to both sides, we see that equation~\eqref{eq:open TRR1,inf} is equivalent to
\begin{align*}
-b_\alpha\frac{\d\mcF_1^o}{\d t^{N+1}_{a+1}}+b_\gamma\eta^{\gamma_\beta}\frac{\d^2\mcF^o_0}{\d t^\beta_0\d t^\alpha_{a+1}}=&\cancel{-\frac{\d^2\mcF^c_0}{\d t^\alpha_a\d t^\mu_0}\eta^{\mu\nu}b_\nu\frac{\d\mcF^o_1}{\d t^{N+1}_0}}+\frac{\d^2\mcF^c_0}{\d t^\alpha_a\d t^\mu_0}\eta^{\mu\nu}\frac{\d^2\mcF^o_0}{\d t^\nu_0\d t^\beta_0}\eta^{\beta\gamma}b_\gamma\\
&-b_\alpha\frac{\d\mcF^o_0}{\d t^{N+1}_a}\frac{\d\mcF^o_1}{\d t_0^{N+1}}+\cancel{b_\gamma\eta^{\gamma\beta}\frac{\d^2\mcF^c_0}{\d t^\beta_0\d t^\alpha_a}\frac{\d\mcF^o_1}{\d t_0^{N+1}}}+\frac{\d\mcF^o_0}{\d t^\alpha_a}\frac{\d^2\mcF^o_0}{\d t_0^{N+1}\d t^\beta_0}\eta^{\beta\gamma}b_\gamma\\
&-\frac{1}{2}b_\alpha\frac{\d^2\mcF_0^o}{\d t^{N+1}_a\d t_0^{N+1}},
\end{align*}
where we adopt the convention $b_{N+1}:=0$. Using the open TRR-$0$ equation and equation~\eqref{eq:open TRR1} we see that this equation is true. Hence, equation~\eqref{eq:open TRR1} is true for any pair $(\mcF^{c,\anc},\mcF^{o,\anc})\in\Anc_N^1$.

\medskip 

Let us prove that equation~\eqref{eq:open TRR1} is preserved by the $G^o_{N+1,+}$-action on $\Anc_N$. We proceed infinitesimally,~i.e we want to prove that equation~\eqref{eq:open TRR1,inf} with $\delta\bullet:=r(z).\bullet$ and $r(z)\in\mfg^o_{N+1,+}$ is satisfied assuming that equation~\eqref{eq:open TRR1} is true.

\medskip

We have
\begin{align*}
\delta\mcF^c_0=&\underbrace{-\sum_{i\ge 1,\,j\ge 0}(r_i)^\mu_\nu q^\nu_j\frac{\d\mcF_0^c}{\d t^\mu_{i+j}}}_{=H \mcF^c_0}+\underbrace{\frac{1}{2}\sum_{i,j\ge 0}(-1)^j(r_{i+j+1})^{\gamma\beta}\frac{\d\mcF_0^c}{\d t^\gamma_i}\frac{\d\mcF_0^c}{\d t^\beta_j}}_{:=Q^c_0},\\
\delta\mcF^o_0=&\underbrace{-\sum_{i\ge 1,\,l\ge 0}(r_i)^\mu_\nu q^\nu_l\frac{\d\mcF_0^o}{\d t^\mu_{i+l}}}_{=H\mcF^o_0}+\underbrace{\sum_{i\ge 1}(-1)^i(r_i)^{N+1,\nu}\frac{\d\mcF_0^c}{\d t^\nu_i}}_{:=L^o_0}+\underbrace{\sum_{j,k\ge 0}(-1)^k(r_{j+k+1})^{\mu\gamma}\frac{\d\mcF_0^o}{\d t^\mu_j}\frac{\d\mcF_0^c}{\d t^\gamma_k}}_{:=Q^o_0},\\
\delta\mcF^o_1=&\underbrace{-\sum_{i\ge 1,\,l\ge 0}(r_i)^\mu_\nu q^\nu_l\frac{\d\mcF_1^o}{\d t^\mu_{i+l}}}_{=H\mcF^o_1}+\underbrace{\sum_{i\ge 1}(-1)^i(r_i)^{N+1,\nu}\frac{\d\mcF_0^o}{\d t^\nu_i}}_{:=L^o_1}\\
&+\underbrace{\frac{1}{2}\sum_{j,k\ge 0}(-1)^k(r_{j+k+1})^{\mu\gamma}\frac{\d\mcF_0^o}{\d t^\mu_j}\frac{\d\mcF_0^o}{\d t^\gamma_k}+\sum_{j,k\ge 0}(-1)^k(r_{j+k+1})^{\mu\gamma}\frac{\d\mcF_1^o}{\d t^\mu_j}\frac{\d\mcF_0^c}{\d t^\gamma_k}}_{:=Q^o_1},
\end{align*}
where $H:=-\sum_{i\ge 1,\,l\ge 0}(r_i)^\mu_\nu q^\nu_l\frac{\d}{\d t^\mu_{i+l}}$.

\medskip

Let us introduce the following operators:
$$
P_{\alpha,a}:=\frac{\d}{\d t^\alpha_{a+1}}-\frac{\d^2\mcF^c_0}{\d t^\alpha_a\d t^\mu_0}\eta^{\mu\nu}\frac{\d}{\d t^\nu_0}-\frac{\d\mcF^o_0}{\d t^\alpha_a}\frac{\d}{\d t_0^{N+1}},\qquad 1\le\alpha\le N+1,\quad a\ge 0.
$$
Then we have $P_{\alpha,a}\frac{\d^2\mcF^c_0}{\d t^\beta_b\d t^\gamma_c}=P_{\alpha,a}\frac{\d\mcF^o_0}{\d t^\beta_b}=0$, $P_{\alpha,a}\frac{\d\mcF^c_0}{\d t^\beta_b}=-\frac{\d^2\mcF^c_0}{\d t^\alpha_a\d t^\beta_{b+1}}$, $P_{\alpha,a}\mcF^o_1=\frac{1}{2}\frac{\d^2\mcF_0^o}{\d t^\alpha_a\d t_0^{N+1}}$, and we have to check that
\begin{align*}
&P_{\alpha,a}\delta\mcF^o_1=\frac{\d^2(\delta\mcF^c_0)}{\d t^\alpha_a\d t^\mu_0}\eta^{\mu\nu}\frac{\d\mcF^o_1}{\d t^\nu_0}+\frac{\d(\delta\mcF^o_0)}{\d t^\alpha_a}\frac{\d\mcF^o_1}{\d t_0^{N+1}}+\frac{1}{2}\frac{\d^2(\delta\mcF_0^o)}{\d t^\alpha_a\d t_0^{N+1}}\quad\Leftrightarrow\\
\Leftrightarrow\quad&P_{\alpha,a}(H\mcF^o_1+L^o_1+Q^o_1)=\frac{\d^2(H\mcF^c_0+Q^c_0)}{\d t^\alpha_a\d t^\mu_0}\eta^{\mu\nu}\frac{\d\mcF^o_1}{\d t^\nu_0}+\frac{\d(H\mcF^o_0+L^o_0+Q^o_0)}{\d t^\alpha_a}\frac{\d\mcF^o_1}{\d t_0^{N+1}}\\
&\hspace{4.3cm}+\frac{1}{2}\frac{\d^2(H\mcF_0^o+L^o_0+Q^o_0)}{\d t^\alpha_a\d t_0^{N+1}}.
\end{align*}
Noting that $P_{\alpha,a}L^o_1=0$ and $\frac{\d L^o_0}{\d t_0^{N+1}}=0$, we collect all the terms with the operator $H$ on the left-hand side and obtain the equivalent equation
\begin{align}\label{eq:openTRR1-inf,equiv1}
&P_{\alpha,a}H\mcF^o_1-\frac{\d^2(H\mcF^c_0)}{\d t^\alpha_a\d t^\mu_0}\eta^{\mu\nu}\frac{\d\mcF^o_1}{\d t^\nu_0}-\frac{\d(H\mcF^o_0)}{\d t^\alpha_a}\frac{\d\mcF^o_1}{\d t_0^{N+1}}-\frac{1}{2}\frac{\d^2(H\mcF_0^o)}{\d t^\alpha_a\d t_0^{N+1}}=\\
=&\frac{\d^2 Q^c_0}{\d t^\alpha_a\d t^\mu_0}\eta^{\mu\nu}\frac{\d\mcF^o_1}{\d t^\nu_0}+\frac{\d(L^o_0+Q^o_0)}{\d t^\alpha_a}\frac{\d\mcF^o_1}{\d t_0^{N+1}}+\frac{1}{2}\frac{\d^2 Q^o_0}{\d t^\alpha_a\d t_0^{N+1}}-P_{\alpha,a}Q^o_1.\notag
\end{align}
Let us compute the left-hand side here.

\medskip

We compute
\begin{align*}
[P_{\alpha,a},H]=&\sum_{i\ge 1}\left(-(r_i)^\mu_\alpha\frac{\d}{\d t^\mu_{i+a+1}}+\frac{\d^2\mcF^c_0}{\d t^\alpha_a\d t^\mu_0}(r_i)^{\gamma\mu}\frac{\d}{\d t^\gamma_i}\right)+\left(H\frac{\d^2\mcF^c_0}{\d t^\alpha_a\d t^\mu_0}\right)\eta^{\mu\nu}\frac{\d}{\d t^\nu_0}+\left(H\frac{\d\mcF^o_0}{\d t^\alpha_a}\right)\frac{\d}{\d t_0^{N+1}},\\
\left[\frac{\d}{\d t^\beta_b},H\right]=&-\sum_{i\ge 1}(r_i)^\mu_\beta\frac{\d}{\d t^\mu_{i+b}}.
\end{align*}
Therefore, the left-hand side of~\eqref{eq:openTRR1-inf,equiv1} is equal to
\begin{align*}
&\frac{1}{2}H\frac{\d^2\mcF_0^o}{\d t^\alpha_a\d t_0^{N+1}}-\sum_{i\ge 1}(r_i)^\mu_\alpha\frac{\d\mcF^o_1}{\d t^\mu_{i+a+1}}+\sum_{i\ge 1}\frac{\d^2\mcF^c_0}{\d t^\alpha_a\d t^\mu_0}(r_i)^{\gamma\mu}\frac{\d\mcF^o_1}{\d t^\gamma_i}+\left(H\frac{\d^2\mcF^c_0}{\d t^\alpha_a\d t^\mu_0}\right)\eta^{\mu\nu}\frac{\d\mcF^o_1}{\d t^\nu_0}\\
&+\left(H\frac{\d\mcF^o_0}{\d t^\alpha_a}\right)\frac{\d\mcF^o_1}{\d t_0^{N+1}}-\left(H\frac{\d^2\mcF^c_0}{\d t^\alpha_a\d t^\mu_0}\right)\eta^{\mu\nu}\frac{\d\mcF^o_1}{\d t^\nu_0}+\sum_{i\ge 1}(r_i)_\alpha^\gamma\frac{\d^2\mcF^c_0}{\d t^\gamma_{i+a}\d t^\mu_0}\eta^{\mu\nu}\frac{\d\mcF^o_1}{\d t^\nu_0}\\
&+\sum_{i\ge 1}\frac{\d^2\mcF^c_0}{\d t^\alpha_a\d t^\gamma_i}(r_i)^\gamma_\mu\eta^{\mu\nu}\frac{\d\mcF^o_1}{\d t^\nu_0}-\left(H\frac{\d\mcF^o_0}{\d t^\alpha_a}\right)\frac{\d\mcF^o_1}{\d t_0^{N+1}}+\sum_{i\ge 1}(r_i)_\alpha^\gamma\frac{\d\mcF^o_0}{\d t^\gamma_{i+a}}\frac{\d\mcF^o_1}{\d t_0^{N+1}}-\frac{1}{2}H\frac{\d^2\mcF_0^o}{\d t^\alpha_a\d t_0^{N+1}}\\
&+\frac{1}{2}\sum_{i\ge 1}(r_i)^\gamma_\alpha\frac{\d^2\mcF_0^o}{\d t^\gamma_{i+a}\d t_0^{N+1}}=\\
=&-\sum_{i\ge 1}(r_i)^\gamma_\alpha\left(\frac{\d\mcF^o_1}{\d t^\gamma_{i+a+1}}-\frac{\d^2\mcF^c_0}{\d t^\gamma_{i+a}\d t^\mu_0}\eta^{\mu\nu}\frac{\d\mcF^o_1}{\d t^\nu_0}-\frac{\d\mcF^o_0}{\d t^\gamma_{i+a}}\frac{\d\mcF^o_1}{\d t_0^{N+1}}-\frac{1}{2}\frac{\d^2\mcF_0^o}{\d t^\gamma_{i+a}\d t_0^{N+1}}\right)\\
&+\sum_{i\ge 1}(-1)^{i+1}\frac{\d^2\mcF^c_0}{\d t^\alpha_a\d t^\mu_0}(r_i)^{\mu\gamma}\frac{\d\mcF^o_1}{\d t^\gamma_i}+\sum_{i\ge 1}\frac{\d^2\mcF^c_0}{\d t^\alpha_a\d t^\gamma_i}(r_i)^\gamma_\mu\eta^{\mu\nu}\frac{\d\mcF^o_1}{\d t^\nu_0}=\\
=&\sum_{i\ge 1}(-1)^{i+1}\frac{\d^2\mcF^c_0}{\d t^\alpha_a\d t^\mu_0}(r_i)^{\mu\gamma}\frac{\d\mcF^o_1}{\d t^\gamma_i}+\sum_{i\ge 1}\frac{\d^2\mcF^c_0}{\d t^\alpha_a\d t^\gamma_i}(r_i)^\gamma_\mu\eta^{\mu\nu}\frac{\d\mcF^o_1}{\d t^\nu_0}.
\end{align*}
Therefore, equation~\eqref{eq:openTRR1-inf,equiv1} is equivalent to
\begin{align}\label{eq:openTRR1-inf,equiv2}
&\sum_{i\ge 1}(-1)^{i+1}\frac{\d^2\mcF^c_0}{\d t^\alpha_a\d t^\mu_0}(r_i)^{\mu\gamma}\frac{\d\mcF^o_1}{\d t^\gamma_i}+\sum_{i\ge 1}\frac{\d^2\mcF^c_0}{\d t^\alpha_a\d t^\gamma_i}(r_i)^\gamma_\mu\eta^{\mu\nu}\frac{\d\mcF^o_1}{\d t^\nu_0}=\\
=&\frac{\d^2 Q^c_0}{\d t^\alpha_a\d t^\mu_0}\eta^{\mu\nu}\frac{\d\mcF^o_1}{\d t^\nu_0}+\frac{\d(L^o_0+Q^o_0)}{\d t^\alpha_a}\frac{\d\mcF^o_1}{\d t_0^{N+1}}+\frac{1}{2}\frac{\d^2 Q^o_0}{\d t^\alpha_a\d t_0^{N+1}}-P_{\alpha,a}Q^o_1.\notag
\end{align}

\medskip

We further compute
\begin{align*}
&\sum_{i\ge 1}\frac{\d^2\mcF^c_0}{\d t^\alpha_a\d t^\gamma_i}(r_i)^\gamma_\mu\eta^{\mu\nu}\frac{\d\mcF^o_1}{\d t^\nu_0}-\frac{\d L^o_0}{\d t^\alpha_a}\frac{\d\mcF^o_1}{\d t_0^{N+1}}=\sum_{i\ge 1}\frac{\d^2\mcF^c_0}{\d t^\alpha_a\d t^\gamma_i}(r_i)^{\gamma\nu}\frac{\d\mcF^o_1}{\d t^\nu_0},\\
&\frac{\d Q^o_0}{\d t^\alpha_a}\frac{\d\mcF^o_1}{\d t_0^{N+1}}=\sum_{j,k\ge 0}(-1)^k(r_{j+k+1})^{\mu\gamma}\left(\underline{\frac{\d\mcF_0^o}{\d t^\alpha_a\d t^\mu_j}\frac{\d\mcF_0^c}{\d t^\gamma_k}\frac{\d\mcF^o_1}{\d t_0^{N+1}}}+\frac{\d\mcF_0^o}{\d t^\mu_j}\frac{\d\mcF_0^c}{\d t^\alpha_a\d t^\gamma_k}\frac{\d\mcF^o_1}{\d t_0^{N+1}}\right),\\
&\frac{\d^2 Q^c_0}{\d t^\alpha_a\d t^\theta_0}\eta^{\theta\nu}\frac{\d\mcF^o_1}{\d t^\nu_0}=\sum_{j,k\ge 0}(-1)^k(r_{j+k+1})^{\mu\gamma}\left(\underline{\underline{\frac{\d^3\mcF_0^c}{\d t^\alpha_a\d t^\theta_0\d t^\mu_j}\frac{\d\mcF_0^c}{\d t^\gamma_k}\eta^{\theta\nu}\frac{\d\mcF^o_1}{\d t^\nu_0}}}+\frac{\d^2\mcF_0^c}{\d t^\alpha_a\d t^\mu_j}\frac{\d^2\mcF_0^c}{\d t^\theta_0\d t^\gamma_k}\eta^{\theta\nu}\frac{\d\mcF^o_1}{\d t^\nu_0}\right),
\end{align*}
and
\begin{align*}
\frac{1}{2}\frac{\d^2 Q^o_0}{\d t^\alpha_a\d t_0^{N+1}}-P_{\alpha,a}Q^o_1=\sum_{j,k\ge 0}(-1)^k(r_{j+k+1})^{\mu\gamma}&\left(\frac{1}{2}\frac{\d^2\mcF^o_0}{\d t^\mu_j\d t_0^{N+1}}\frac{\d^2\mcF^c_0}{\d t^\gamma_k\d t^\alpha_a}-\underline{\underline{\frac{\d^3\mcF^c_0}{\d t^\alpha_a\d t^\mu_j\d t^\theta_0}\eta^{\theta\nu}\frac{\d\mcF^o_1}{\d t^\nu_0}\frac{\d\mcF^c_0}{\d t^\gamma_k}}}\right.\\
&\left.\hspace{0.3cm}-\underline{\frac{\d^2\mcF^o_0}{\d t^\alpha_a\d t^\mu_j}\frac{\d\mcF^o_1}{\d t_0^{N+1}}\frac{\d\mcF^c_0}{\d t^\gamma_k}}+\frac{\d\mcF^o_1}{\d t^\mu_j}\frac{\d^2\mcF^c_0}{\d t^\alpha_a\d t^\gamma_{k+1}}\right),
\end{align*}
where we underlined the terms that cancel each other when we substitute these expressions on the right-hand side of~\eqref{eq:openTRR1-inf,equiv2}. So equation~\eqref{eq:openTRR1-inf,equiv2} is equivalent to
\begin{align}\label{eq:openTRR1-inf,equiv3}
&\sum_{i\ge 1}(r_i)^{\mu\gamma}\left((-1)^{i+1}\frac{\d^2\mcF^c_0}{\d t^\alpha_a\d t^\mu_0}\frac{\d\mcF^o_1}{\d t^\gamma_i}+\frac{\d^2\mcF^c_0}{\d t^\alpha_a\d t^\mu_i}\frac{\d\mcF^o_1}{\d t^\gamma_0}\right)-\sum_{j,k\ge 0}(-1)^k(r_{j+k+1})^{\mu\gamma}\frac{\d\mcF^o_1}{\d t^\mu_j}\frac{\d^2\mcF^c_0}{\d t^\alpha_a\d t^\gamma_{k+1}}=\\
=&\sum_{j,k\ge 0}(-1)^k(r_{j+k+1})^{\mu\gamma}\left(\frac{\d\mcF_0^o}{\d t^\mu_j}\frac{\d\mcF^o_1}{\d t_0^{N+1}}+\frac{1}{2}\frac{\d^2\mcF^o_0}{\d t^\mu_j\d t_0^{N+1}}+\frac{\d^2\mcF^c_0}{\d t^\theta_0\d t^\mu_j}\eta^{\theta\nu}\frac{\d\mcF^o_1}{\d t^\nu_0}\right)\frac{\d^2\mcF^c_0}{\d t^\gamma_k\d t^\alpha_a}.\notag
\end{align}
Finally note that the left-hand side of~\eqref{eq:openTRR1-inf,equiv3} is equal exactly to $\sum_{j,k\ge 0}(-1)^k(r_{j+k+1})^{\mu\gamma}\frac{\d\mcF^o_1}{\d t^\mu_{j+1}}\frac{\d^2\mcF^c_0}{\d t^\alpha_a\d t^\gamma_k}$, and using the open TRR-$1$ equations we conclude that equation~\eqref{eq:openTRR1-inf,equiv3} is true. Therefore, equation~\eqref{eq:open TRR1} is true for any pair of total ancestor potentials.

\medskip

Any pair of total descendant potentials is obtained from a pair of total ancestor potentials by the action of an element of the group $G^o_{N+1,-}$. Since we have already proved that equation~\eqref{eq:open TRR1} is preserved by the $G^o_{N+1,-}$, equation~\eqref{eq:open TRR1} is true for any pair of total descendant potentials.

\medskip

\subsection{The open string equation}

Consider the following operator:
\[
 \cL_{-1}^{\eta,\oA}:=\frac{\eps^{-2}}{2}\eta_{\alpha\beta} q_0^\alpha q_0^\beta + \eps^{-1} q_0^{N+1}+\sum_{d\ge 1} q^\alpha_d \frac{\p}{\p q^\alpha_{d-1}}.
\]
Using the closed string equation~\eqref{eq:closed string equation}, we see that equation~\eqref{eq:open string equation} follows from the equation
\begin{gather}\label{eq:full open string}
\lb\cL_{-1}^{\eta,\oA}+\frac{1}{2}\mathrm{tr}(r_1)\rb\exp(\cF^{c,\desc}+\cF^{o,\desc})=0,
\end{gather}
which we are going to prove.

\medskip

\begin{lemma}\label{lemma:commutation relations with Lm1}
We have the following commutation relations.
\begin{enumerate}[ 1.]
\item $\left[\widehat{s(z)}^o,\mcL^{\eta,\oA}_{-1}\right]=0$ for any $s(z)\in\mfg^o_{N+1,-}$.

\smallskip

\item $\left[\widehat{\psi},\cL^{\eta,\oA}_{-1}\right] = A^\alpha\psi_\alpha^\beta\frac{\d}{\d t^\beta_0}+\frac{\eps^{-2}}{2}(\psi^\gamma_\alpha\eta_{\gamma\beta}+\eta_{\alpha\gamma}\psi^\gamma_\beta)t^\alpha_0 t^\beta_0$ for any $\psi\in\End(\mbC^N)$.

\smallskip

\item $\left[\widehat{\ob},\cL^{\eta,\oA}_{-1}\right] = b_\alpha A^\alpha\frac{\d}{\d t^{N+1}_0}$ for any $\ob\in\mbC^N$.

\smallskip

\item $\left[\widehat{r(z)}^o,\cL^{\eta,\oA}_{-1}\right] = \frac{1}{2} \mathrm{tr}(r_1)$ for any $r(z)\in\mfg^o_{N+1,+}$.
\end{enumerate}
\end{lemma}
\begin{proof}
First of all, note that $\cL^{\eta,\oA}_{-1}=-\widehat{\Id\cdot z^{-1}}^o$, which immediately implies Part 1 of the lemma.

\medskip

Parts 2 and 3 are simple direct computations.

\medskip

For Part 4, let us decompose $\cL_{-1}^{\eta,\oA} = A + B$, where $A:=\frac{\eps^{-2}}{2}\eta_{\alpha\beta} q_0^\alpha q_0^\beta + \eps^{-1} q_0^{N+1}$ and $B:=\sum_{d\ge 1} q^\alpha_d \frac{\p}{\p q^\alpha_{d-1}}$. We compute
\begin{align*}
&\left[\widehat{r(z)}^o,A\right]=\sum_{i \ge 0} (r_{i+1})^\alpha_\nu q_0 ^\nu \frac{\p}{\p q^\alpha_i} + \eps \sum_{i \ge 0} (-1)^i(r_{i+1})^{N+1,\alpha}\frac{\p}{\p q^\alpha_i}+\frac{1}{2} \mathrm{tr}(r_1),\\
&\left[-\sum_{i\ge 1,\,j\ge 0}(r_i)^\mu_\nu q^\nu_j\frac{\d}{\d q^\mu_{i+j}}+\eps\sum_{i\ge 1}(-1)^i(r_i)^{N+1,\nu}\frac{\d}{\d q^\nu_i},B\right]=-\sum_{k\ge 1}(r_k)^\mu_\nu q^\nu_0 \frac{\p}{\p q^\mu_{k-1}}\\
&\hspace{9.9cm}+\eps \sum_{i \ge 1} (-1)^i (r_i)^{N+1,\nu} \frac{\p}{\p q^\nu_{i-1}}, \\
&\left[\frac{\eps^2}{2}\sum_{i,j\ge 0}(-1)^j(r_{i+j+1})^{\alpha\beta}\frac{\d^2}{\d q^\alpha_i\d q^\beta_j},B\right]=0,
\end{align*}
which completes the proof of the lemma.
\end{proof}

\medskip

By~\cite[Theorem~1.3]{BT17}, equation~\eqref{eq:full open string} is true for $\mcF^o=\mcF^\PST(t^1_*,t^{N+1}_*,\eps)$ and $\mcF^c=\sum_{i=1}^N\mcF^\KW(t^i_*,\eps)$. Using Part~1 of Lemma~\ref{lemma:commutation relations with Lm1}, it is easy to see that equation~\eqref{eq:full open string} is true for $\mcF^o=\mcF^\PST_{(\theta)}(a_1 t^1_*,t^{N+1}_*,\eps)$ and $\mcF^c=\sum_{i=1}^N\mcF^\KW(a_i t^i_*,a_i\eps)$ for any $a_1,\ldots,a_N,\theta\in\mbC^*$.

\medskip

Part 2 of Lemma~\ref{lemma:commutation relations with Lm1} implies that for any $\psi\in\End(\mbC^N)$ we have
\[
\exp\lb\widehat\psi\rb\circ\cL^{(\delta_{ij}),(a_1^{-1},\ldots,a_N^{-1},0)}_{-1}=\cL_{-1}^{\eta,\oA}\circ\exp\lb\widehat\psi\rb,
\]
where $\eta_{\alpha\beta}=\sum_{i=1}^N\exp(\psi)^i_\alpha\exp(\psi)^i_\beta$ and $\oA=(A^1,\ldots,A^{N+1})$ is given by $A^\alpha=\sum_{i=1}^N\exp(-\psi)^\alpha_i a_i^{-1}$ for $1\le\alpha\le N$, and $A^{N+1}=0$. This proves equation~\eqref{eq:full open string} for any pair of total ancestor potentials from the space $\Anc^0_N$.

\medskip

Part 3 of Lemma~\ref{lemma:commutation relations with Lm1} implies that 
$$
\exp\lb\widehat\ob\rb\circ\cL_{-1}^{\eta,\oA}=\cL_{-1}^{\eta,(A^1,\ldots,A^N,A^{N+1}-b_\alpha A^\alpha)}\circ\exp\lb\widehat\ob\rb,
$$
which proves equation~\eqref{eq:full open string} for any pair of total ancestor potentials from the space $\Anc^1_N$.

\medskip

Part 4 of Lemma~\ref{lemma:commutation relations with Lm1} implies that 
$$
\exp\lb\widehat{r(z)}^o\rb\circ\cL_{-1}^{\eta,\oA}=\left(\cL_{-1}^{\eta,\oA}+\frac{1}{2}\mathrm{tr}(r_1)\right)\circ\exp\lb\widehat{r(z)}^o\rb,\quad r(z)\in\mfg^o_{N+1,+},
$$
which proves equation~\eqref{eq:full open string} for any pair of total ancestor potentials. Finally, Part~1 of Lemma~\ref{lemma:commutation relations with Lm1} shows that equation~\eqref{eq:full open string} is true for any pair of total descendant potentials. 

\medskip

\subsection{The open dilaton equation}

Consider the following operator:
\[
 \mcP^{\oA}:=\sum_{k\ge 0} q^\alpha_k \frac{\p}{\p q^\alpha_k}+\eps\frac{\p}{\p \eps}+\frac{N}{24}+\frac{1}{2}.
\]
Using the closed dilaton equation~\eqref{eq:closed dilaton equation}, we see that equation~\eqref{eq:open dilaton equation} follows from the equation
\begin{gather}\label{eq:full open dilaton}
\mcP^{\oA}\exp(\cF^{c,\desc} + \cF^{o,\desc}) = 0,
\end{gather}
which we are going to prove.

\medskip

\begin{lemma}
We have the following commutation relations.
\begin{enumerate}[ 1.]
\item $\left[\widehat{s(z)}^o,\mcP^{\oA}\right]=\left[\widehat{r(z)}^o,\mcP^{\oA}\right]=0$ for any $s(z)\in\mfg^o_{N+1,-}$ and $r(z)\in\mfg^o_{N+1,+}$.

\smallskip

\item $\left[\widehat{\psi},\mcP^{\oA}\right] = A^\alpha\psi_\alpha^\beta\frac{\d}{\d t^\beta_1}$ for any $\psi\in\End(\mbC^N)$.

\smallskip

\item $\left[\widehat{\ob},\mcP^{\oA}\right] = b_\alpha A^\alpha\frac{\d}{\d t^{N+1}_1}$ for any $\ob\in\mbC^N$.
\end{enumerate}
\end{lemma}
\begin{proof}
Part 1 of the lemma is obvious. Parts 2 and 3 are simple direct computations.
\end{proof}

\medskip

Using this lemma, the proof follows exactly in the same way as for the open string equation starting from the fact that equation~\eqref{eq:full open dilaton} is true for $\mcF^o=\mcF^\PST(t^1_*,t^{N+1}_*,\eps)$ and $\mcF^c=\sum_{i=1}^N\mcF^\KW(t^i_*,\eps)$~\cite[Theorem~1.3]{BT17}.

\medskip

\section{Further properties of the construction and an example in genus $1$}\label{section:more properties}

In this section, we present more formulas for the actions of the groups $G^o_{N+1,+}$ and $\mbC^N$ on the spaces $\Anc_N$ and $\Anc_N^1$, respectively (see Proposition~\ref{proposition:formula for openr} and~Lemma~\ref{lemma:formula for expb}), and then use them for the computation of the coefficient of $t^{N+1}_0$ in an arbitrary open ancestor potential in genus $1$ (see Proposition~\ref{proposition:open correlator in genus one}).

\medskip

\subsection{On the $G^o_{N+1,+}$-action on the space $\Anc_N$}

Given $R(z)\in G^o_{N+1,+}$, let $r(z):=\log R(z)$ and 
\begin{align*}
&L_1:=-\sum_{i\ge 1,\,j\ge 0}(r_i)^\mu_\nu q^\nu_j\frac{\d}{\d q^\mu_{i+j}},\qquad L_2:=\eps\sum_{\substack{1\le\alpha\le N\\a\ge 1}} L_2^{\alpha,a}\frac{\d}{\d q^\alpha_a},\qquad L_3:=\eps\sum_{a\ge 2}L_3^a\frac{\d}{\d q^{N+1}_a},\\
&E:=\frac{\eps^2}{2}\sum_{a,b\ge 0}E^{\alpha,a;\beta,b}\frac{\d^2}{\d q^\alpha_a\d q^\beta_b},
\end{align*}
where
\begin{align*}
&L_2^{\alpha,a}:=\Coef_{z^a}R(-z)^{N+1}_\mu\eta^{\mu\alpha},&& 1\le\alpha\le N,\quad a\ge 1,\\
&L_3^a:=-\frac{1}{2}\sum_{i+j=a}{a\choose i}\left((R^{-1})_i\right)^{N+1}_\mu\eta^{\mu\nu}\left((R^{-1})_j\right)^{N+1}_\nu,&& a\ge 2,\\
&E^{\alpha,a;\beta,b}:=
\begin{cases}
\Coef_{z_1^a z_2^b}\frac{\left(1-R^{-1}(z_1)R(-z_2)\right)^\alpha_\mu\eta^{\mu\beta}}{z_1+z_2},&\text{if $\beta\le N$},\\
E^{N+1,b;\alpha,a},&\text{if $\alpha\le N$ and $\beta=N+1$},\\
0,&\text{if $\alpha=\beta=N+1$}.
\end{cases}
\end{align*}

\medskip

\begin{proposition}\label{proposition:formula for openr}
For an arbitrary pair $(\mcF^{c,\anc},\mcF^{o,\anc})$ of closed and open total ancestor potentials we have
$$
\exp\left(\widehat{r(z)}^o\right)\exp(\mcF^{c,\anc}+\mcF^{o,\anc})=\exp(L_1)\exp(L_2)\exp(L_3)\exp(E)\exp(\mcF^{c,\anc}+\mcF^{o,\anc}).
$$
\end{proposition}
\begin{proof}
Let us apply part (b) of Lemma~\ref{lemma:BCH-special} with
$$
X:=-\sum_{i\ge 1,\,j\ge 0}(r_i)^\mu_\nu q^\nu_j\frac{\d}{\d q^\mu_{i+j}}+\eps\sum_{i\ge 1}(-1)^i(r_i)^{N+1,\nu}\frac{\d}{\d q^\nu_i},\qquad Y:=\frac{\eps^2}{2}\sum_{i,j\ge 0}(-1)^j(r_{i+j+1})^{\alpha\beta}\frac{\d^2}{\d q^\alpha_i\d q^\beta_j}.
$$
After a long and tedious computation, using Lemma~\ref{lemma:technical lemma}, we obtain
$$
\sum_{n\ge 0}\frac{(-1)^n}{(n+1)!}\ad_X^n Y=\frac{\eps^2}{2}\sum_{a,b\ge 0}\tE^{\alpha,a;\beta,b}\frac{\d^2}{\d q^\alpha_a\d q^\beta_b},
$$
where $\tE^{\alpha,a;\beta,b}=E^{\alpha,a;\beta,b}$ if $\alpha\le N$ or $\beta\le N$, and 
\begin{align*}
\tE^{N+1,a;N+1,b}=\Coef_{z_1^a z_2^b}\scriptstyle\frac{\left(\frac{R(-z_1)-1}{r(-z_1)}\right)^{N+1}_\mu\eta^{\mu\nu}r(z_1)^{N+1}_\nu-R^{-1}(z_1)^{N+1}_\mu\eta^{\mu\nu}R^{-1}(z_2)^{N+1}_\nu+r(z_2)^{N+1}_\mu\eta^{\mu\nu}\left(\frac{R(-z_2)-1}{r(-z_2)}\right)^{N+1}_\nu}{z_1+z_2}.
\end{align*}
We therefore obtain
\begin{align*}
\exp\left(\widehat{r(z)}^o\right)=&\exp\left(-\sum_{i\ge 1,\,j\ge 0}(r_i)^\mu_\nu q^\nu_j\frac{\d}{\d q^\mu_{i+j}}+\eps\sum_{i\ge 1}(-1)^i(r_i)^{N+1,\nu}\frac{\d}{\d q^\nu_i}\right)\circ\\
&\circ\exp\left(\frac{\eps^2}{2}\sum_{a,b\ge 0}\tE^{N+1,a;N+1,b}\frac{\d^2}{\d q^{N+1}_a\d q^{N+1}_b}\right)\exp(E).
\end{align*}

\medskip

We then again use part (b) of Lemma~\ref{lemma:BCH-special} with $X:=-\sum_{i\ge 1,\,j\ge 0}(r_i)^\mu_\nu q^\nu_j\frac{\d}{\d q^\mu_{i+j}}$ and $Y:=\eps\sum_{i\ge 1}(-1)^i(r_i)^{N+1,\nu}\frac{\d}{\d q^\nu_i}$ and obtain
\begin{align*}
&\exp\left(-\sum_{i\ge 1,\,j\ge 0}(r_i)^\mu_\nu q^\nu_j\frac{\d}{\d q^\mu_{i+j}}+\eps\sum_{i\ge 1}(-1)^i(r_i)^{N+1,\nu}\frac{\d}{\d q^\nu_i}\right)=\\
=&\exp(L_1)\exp(L_2)\exp\left(-\eps\sum_{a\ge 0}\Coef_{z^a}\left[\left(\frac{R(-z)-1}{r(-z)}\right)^{N+1}_\mu\eta^{\mu\nu}r(z)^{N+1}_\nu\right]\frac{\d}{\d q^{N+1}_a}\right).
\end{align*}

\medskip

Recall that $\frac{\d\exp(\mcF^{o,\anc})}{\d q^{N+1}_a}=\frac{\eps^a}{(a+1)!}\frac{\d^{a+1}\exp(\mcF^{o,\anc})}{(\d q^{N+1}_0)^{a+1}}$, which gives $\frac{\d^2\exp(\mcF^{o,\anc})}{\d q^{N+1}_i\d q^{N+1}_j}=\eps^{-1}{i+j+2\choose i+1}\frac{\d\exp(\mcF^{o,\anc})}{\d q^{N+1}_{i+j+1}}$. Therefore, it remains to check that
$$
-\sum_{a\ge 0}\Coef_{z^a}\left[\left(\frac{R(-z)-1}{r(-z)}\right)^{N+1}_\mu\eta^{\mu\nu}r(z)^{N+1}_\nu\right]+\frac{1}{2}\sum_{i+j=a-1}{a+1\choose i+1}\tE^{N+1,i;N+1,j}=L_3^a,
$$
which is a simple direct computation based on the following observation:
$$
\sum_{i+j=a}{a+2\choose i+1}\Coef_{z_1^iz_2^j}P(z_1,z_2)=\sum_{i+j=a+1}{a+1\choose i}\Coef_{z_1^iz_2^j}\left[(z_1+z_2)P(z_1,z_2)\right],
$$
where $P(z_1,z_2)\in\mbC[z_1,z_2]$.
\end{proof}

\medskip

\subsection{On the $\mbC^N$-action on the space $\Anc^1_N$}

Recall that given $\ob=(b_1,\ldots,b_N)\in\mbC^N$ the operator $\widehat{\ob}$ is defined by
$$
\widehat{\ob}=\sum_{d\ge 0}b_\alpha t^\alpha_d\frac{\d}{\d t^{N+1}_d}-\eps b^\alpha\frac{\d}{\d t^\alpha_0},
$$
where we introduced the notation $b^\alpha:=\eta^{\alpha\beta}b_\beta$.

\medskip

\begin{lemma}\label{lemma:formula for expb}
We have
$$
\exp\left(\widehat{\ob}\right)=\exp\left(\sum_{d\ge 0}b_\alpha t^\alpha_d\frac{\d}{\d t^{N+1}_d}\right)\exp\left(-\eps b^\alpha\frac{\d}{\d t^\alpha_0}-\eps b^\alpha b_\alpha\frac{\d}{\d t^{N+1}_0}\right).
$$
\end{lemma}
\begin{proof}
Easy application of part (b) of Lemma~\ref{lemma:BCH-special} with $X:=\sum_{d\ge 0}b_\alpha t^\alpha_d\frac{\d}{\d t^{N+1}_d}$ and $Y:=-\eps b^\alpha\frac{\d}{\d t^\alpha_0}$.
\end{proof}

\medskip

\subsection{An example in genus $1$}

We consider an arbitrary pair of closed and open total ancestor potentials $(\mcF^{c,\anc},\mcF^{o,\anc})$ corresponding to parameters $a_1,\ldots,a_N,\theta\in\mbC^*$, $\ob=(b_1,\ldots,b_N)\in\mbC^N$, $\psi\in\End(\mbC^N)$, and $r(z)\in\mfg^o_{N+1,+}$:
\begin{align*}
&\exp(\mcF^{c,\anc}):=\exp\left(\widehat{\pi(r(z))}^c\right)\exp\left(\widehat{\psi}\right)\left(\prod_{i=1}^N\tau^{\KW}(a_i t^i_*,a_i\eps)\right),\\
&\exp(\mcF^{c,\anc}+\mcF^{o,\anc}):=\exp\left(\widehat{r(z)}^o\right)\exp\left(\widehat{\ob}\right)\exp\left(\widehat{\psi}\right)\left(\tau^{\PST}_{(\theta)}(a_1 t^1_*,t^{N+1}_*,\eps)\prod_{i=1}^N\tau^{\KW}(a_i t^i_*,a_i\eps)\right).
\end{align*}
Recall that the $N\times N$ matrix $(\eta^{\alpha\beta})$, used in the expressions for $\widehat{\pi(r(z))}^c$, $\widehat{r(z)}^o$, and $\widehat{\ob}$, is given by $\eta^{\alpha\beta}=\sum_{i=1}^N(\Psi^{-1})^\alpha_i(\Psi^{-1})^\beta_i$, where $\Psi=(\Psi^i_\alpha)=\exp(\psi)$. The corresponding solutions of the closed and open WDVV equations admit a unit given by $A^\alpha=\begin{cases}\sum_{i=1}^N(\Psi^{-1})^\alpha_i a_i^{-1},&\text{if $1\le\alpha\le N$},\\-\sum_{\beta=1}^N b_\beta A^\beta,&\text{$\alpha=N+1$}.\end{cases}$

\medskip

For an $(N+1)\times(N+1)$ matrix $M=(M^\alpha_\beta)$, we will denote $M^\alpha_\un:=M^\alpha_\beta A^\beta$.

\medskip

\begin{proposition}\label{proposition:open correlator in genus one}
We have
\begin{enumerate}[ 1.]
\item $\displaystyle\left.\mcF^{o,\anc}_1\right|_{t^*_*=0}=0$,

\smallskip

\item $\displaystyle \Coef_{t^{N+1}_0}\mcF^{o,\anc}_1=\frac{\theta}{2}(r_1)^\alpha_\un(a_1\Psi^1_\alpha+\theta b_\alpha)+\frac{\theta^2}{2}(r_1)^{N+1}_\un-a_1\Psi^1_\alpha b^\alpha-\theta b_\alpha b^\alpha$.
\end{enumerate}
\end{proposition}
\begin{proof}
Consider the total ancestor potentials $\tmcF^{c,\anc}$ and $\tmcF^{o,\anc}$ given by
\begin{align*}
&\exp(\tmcF^{c,\anc}):=\exp\left(\widehat{\psi}\right)\left(\prod_{i=1}^N\tau^{\KW}(a_i t^i_*,a_i\eps)\right),\\
&\exp(\tmcF^{c,\anc}+\tmcF^{o,\anc}):=\exp\left(\widehat{\ob}\right)\exp\left(\widehat{\psi}\right)\left(\tau^{\PST}_{(\theta)}(a_1 t^1_*,t^{N+1}_*,\eps)\prod_{i=1}^N\tau^{\KW}(a_i t^i_*,a_i\eps)\right).
\end{align*}
Note that the operator $\exp(L_1)$ acts as the substitution $t^\alpha_d\mapsto t^\alpha_d+\sum_{i=1}^d((R^{-1})_i)^\alpha_\beta(t^\beta_{d-i}-\delta_{d-i,1}A^\beta)$. We therefore introduce an operator $L_1'$ by
$$
L'_1:=-\sum_{d\ge 2}((R^{-1})_{d-1})^\alpha_\un\frac{\d}{\d t^\alpha_d}.
$$

\medskip

\begin{lemma}
We have
\begin{align*}
&\mcF^{o,\anc}_1|_{t^{\le N}_0=t^*_{\ge 1}=0}=\\
=&\left.\exp\left(L_1'\right)\left(\tmcF^{o,\anc}_1+\sum_{a,b\ge 0}\frac{E^{\alpha,a;\beta,b}}{2}\frac{\d\tmcF^{o,\anc}_0}{\d t^\alpha_a}\frac{\d\tmcF^{o,\anc}_0}{\d t^\beta_b}+\sum_{a\ge 1}L_2^{\alpha,a}\frac{\d\tmcF^{o,\anc}_0}{\d t^\alpha_a}+\sum_{a\ge 2}L_3^a\frac{\d\tmcF^{o,\anc}_0}{\d t^{N+1}_a}\right)\right|_{t^{\le N}_0=t^*_{\ge 1}=0}.
\end{align*}
\end{lemma}
\begin{proof}
This follows from Proposition~\ref{proposition:formula for openr}, the property $\frac{\d\tmcF^{c,\anc}_0}{\d t^{N+1}_0}=0$, and the fact that (see Section~\ref{subsubsection:closed descendant potentials})
\begin{gather}\label{eq:closed vanishing in genus 0}
\text{the coefficient of $t^{\alpha_1}_{d_1}\ldots t^{\alpha_n}_{d_n}$ in $\tmcF^{c,\anc}_0$ is zero if $\sum d_i\ge n-2$}.
\end{gather}
\end{proof}

\medskip

\begin{lemma}
For $k=0$ or $k=1$, we have
$$
\Coef_{(t^{N+1}_0)^k}\mcF^{o,\anc}_1=\Coef_{(t^{N+1}_0)^k}\left(\exp\left(L_1'\right)\tmcF^{o,\anc}_1\right).
$$
\end{lemma}
\begin{proof}
This follows from the previous lemma and the fact that (see Section~\ref{subsubsection:definition of open descendant})
\begin{gather}\label{eq:vanishing in genus 0}
\text{the coefficient of $t^{\alpha_1}_{d_1}\ldots t^{\alpha_n}_{d_n}$ in~$\tmcF^{o,\anc}_0$ is zero if $\sum d_i\ge n-1$.}
\end{gather}
\end{proof}

\medskip

Define the \emph{open $G$-function} corresponding to $\mcF^{o,\anc}_1$ by  
$$
G^o(t^1,\ldots,t^{N+1}):=\left.\mcF^{o,\anc}_1\right|_{t^*_{\ge 1}=0}\in\mbC[[t^1,\ldots,t^{N+1}]].
$$
Denote $v^\alpha:=\eta^{\alpha\beta}\frac{\d^2\mcF^{c,\anc}_0}{\d t^\beta_0\d t^\un_0}$, $1\le\alpha\le N$, and $\phi:=\frac{\d\mcF^{o,\anc}_0}{\d t^\un_0}$. Note that the closed and open string equations imply that 
\begin{gather*}
v^\alpha|_{t^*_{\ge 1}=0}=t^\alpha_0,\quad 1\le\alpha\le N,\qquad\qquad \phi|_{t^*_{\ge 1}=0}=t^{N+1}_0.
\end{gather*}
In~\cite{BB21b}, the authors proved that 
$$
\mcF^{o,\anc}_1=G^o(v^1,\ldots,v^N,\phi)+\frac{1}{2}\log\frac{\d^2\mcF^{o,\anc}_0}{\d t^\un_0\d t^{N+1}_0}.
$$

\medskip

Denote by $\tG^o$ the open $G$-function corresponding to~$\tmcF^{o,\anc}_1$.

\medskip

\begin{lemma}
We have
\begin{enumerate}[ 1.]
\item $\displaystyle\mcF^{o,\anc}_1|_{t^*_*=0}=\tG^o|_{t^*_*=0}$,

\smallskip

\item $\displaystyle\Coef_{t^{N+1}_0}\mcF^{o,\anc}_1=\Coef_{t^{N+1}_0}\tG^o+\frac{1}{2}\left.(r_1)^\alpha_\un\frac{\d^4\tmcF^{o,\anc}_0}{\d t^\un_0(\d t^{N+1}_0)^2\d t^\alpha_2}\right|_{t^*_*=0}$.
\end{enumerate}
\end{lemma}
\begin{proof}
Combining the previous lemma, the formula
$$
\tmcF^{o,\anc}_1=\tG^o(\tv^1,\ldots,\tv^N,\tphi)+\frac{1}{2}\log\frac{\d^2\tmcF^{o,\anc}_0}{\d t^\un_0\d t^{N+1}_0},
$$
the property~\eqref{eq:closed vanishing in genus 0}, and the vanishing~$\frac{\d\tv^\alpha}{\d t^{N+1}_0}=0$, we obtain
$$
\Coef_{(t^{N+1}_0)^k}\mcF^{o,\anc}_1=\Coef_{(t^{N+1}_0)^k}\left[\tG^o(0,\ldots,0,\exp(L_1')\phi)+\frac{1}{2}\log\left(\exp(L_1')\frac{\d^2\tmcF^{o,\anc}_0}{\d t^\un_0\d t^{N+1}_0}\right)\right]
$$
for $k=0$ or for $k=1$. 

\medskip

For part 1, using the property~\eqref{eq:vanishing in genus 0}, we obtain $\left.\exp(L_1')\phi\right|_{t^*_*=0}=0$ and $\left.\exp(L_1')\frac{\d^2\tmcF^{o,\anc}_0}{\d t^\un_0\d t^{N+1}_0}\right|_{t^*_*=0}=1$, which immediately proves part 1.

\medskip

For part 2, we again use the property~\eqref{eq:vanishing in genus 0} and obtain $\left.\frac{\d}{\d t^{N+1}_0}\exp(L_1')\phi\right|_{t^*_*=0}=1$ and
$$
\left.\frac{\d}{\d t^{N+1}_0}\log\left(\exp(L_1')\frac{\d^2\tmcF^{o,\anc}_0}{\d t^\un_0\d t^{N+1}_0}\right)\right|_{t^*_*=0}=\left.\exp(L_1')\frac{\frac{\d^3\tmcF^{o,\anc}_0}{\d t^\un_0(\d t^{N+1}_0)^2}}{\frac{\d^2\tmcF^{o,\anc}_0}{\d t^\un_0\d t^{N+1}_0}}\right|_{t^*_*=0}=\left.(r_1)^\alpha_\un\frac{\d^4\tmcF^{o,\anc}_0}{\d t^\un_0(\d t^{N+1}_0)^2\d t^\alpha_2}\right|_{t^*_*=0},
$$
which completes the proof.
\end{proof}

\medskip

Using the open TRR-$0$ equations, the fact that 
\begin{align*}
&\left.\tmcF^{c,\anc}_0\right|_{t^*_{\ge 1}=0}=\sum_{i=1}^N a_i\frac{(\Psi^i_\gamma t^\gamma_0)^3}{6},\\
&\left.\tmcF^{o,\anc}_0\right|_{t^*_{\ge 1}=0}=a_1(\Psi^1_\alpha t^\alpha_0)(t^{N+1}_0+b_\beta t^\beta_0)+\theta\frac{(t^{N+1}_0+b_\beta t^\beta_0)^2}{2}+\frac{(t^{N+1}_0+b_\beta t^\beta_0)^3}{6}-\sum_{i=1}^N b^\alpha\Psi^i_\alpha\frac{(\Psi^i_\gamma t^\gamma_0)^2}{2},
\end{align*}
and the open string equation, it is easy to compute that
$$
\left.\frac{\d^4\tmcF^{o,\anc}_0}{\d t^\un_0(\d t^{N+1}_0)^2\d t^\alpha_2}\right|_{t^*_*=0}=
\begin{cases}
\theta(a_1\Psi^1_\alpha+\theta b_\alpha),&\text{if $1\le\alpha\le N$},\\
\theta^2,&\text{if $\alpha=N+1$}.
\end{cases}
$$

\medskip

Therefore, the following lemma completes the proof of the proposition.

\begin{lemma}
We have
\begin{align*}
\tG^o=&\sum_{i=1}^N a_i\frac{(\Psi^i_\alpha b^\alpha)^2}{2}\Psi^i_\gamma t^\gamma-a_1(\Psi^1_\alpha b^\alpha)(t^{N+1}+b_\gamma t^\gamma)\\
&-b_\alpha b^\alpha\left(a_1\Psi^1_\gamma t^\gamma+\theta(t^{N+1}+b_\gamma t^\gamma)+\frac{1}{2}(t^{N+1}+b_\gamma t^\gamma)^2\right).
\end{align*}
\end{lemma}
\begin{proof}
By~\eqref{eq:dimension condition for PST}, we have $\left.\mcF^{\PST}_1\right|_{t_{\ge 1}=s_{\ge 1}=0}=0$. Equation~\eqref{eq:shifting by theta} implies that $\left.\mcF^{\PST}_{(\theta);1}\right|_{t_{\ge 1}=s_{\ge 1}=0}=0$. Then the lemma is proved by a careful application of Lemma~\ref{lemma:formula for expb}.
\end{proof}
\end{proof}

\medskip

{\appendix

\section{Technical lemmas}

\begin{lemma}\label{lemma:BCH-special}
Let $X$ and $Y$ be two operators such that $[Y,\ad^n_X Y]=0$ for any $n\ge 0$. Consider also a formal variable $z$. Then we have
\begin{enumerate}[ (a)]
\item $\displaystyle\exp(z(X+Y))=\exp\left(\sum_{n\ge 0}\frac{z^{n+1}}{(n+1)!}\ad^n_X Y\right)\exp(z X)$,

\smallskip

\item $\displaystyle\exp(z(X+Y))=\exp(zX)\exp\left(\sum_{n\ge 0}\frac{(-1)^n z^{n+1}}{(n+1)!}\ad_X^n Y\right)$.
\end{enumerate}
\end{lemma}
\begin{proof}
Both formulas are well-known special cases of the Baker--Campbell--Hausdorff formula. Note that the second formula can be obtained from the first one by changing $X\mapsto X-Y$ and then $X\mapsto -X$.
\end{proof}

\medskip

\begin{lemma}\label{lemma:technical lemma}
For any two operators $A,B$ and a formal variable $z$ we have
\begin{align*}
&\sum_{m,n\ge 0}z^{m+n+1}\frac{A^m(A+B)B^n}{m!n!(m+n+1)}=\exp(zA)\exp(zB)-1,\\
&\sum_{m,n\ge 0}z^{m+n+1}\frac{A^m(A+B)B^n}{m!(n+1)!(m+n+2)}=\exp(zA)\frac{\exp(zB)-1}{zB}-\frac{\exp(zA)-1}{zA}.
\end{align*}
\end{lemma}
\begin{proof}
Elementary exercise.
\end{proof}
}

\end{document}